\documentclass[10pt,aps,twocolumn,superscriptaddress,longbibliography]{revtex4-1}
\usepackage{amsmath, amsfonts, amssymb, enumerate, amsthm, graphicx, braket, bbm, mathrsfs, tikz, paralist, subfigure}

\usepackage[colorlinks=true,linkcolor=blue,citecolor=red, linktocpage=true,breaklinks=true]{hyperref}
\usepackage{times}

\usepackage{soul}

\newtheorem{theo}{Theorem}
\newtheorem{thm}[theo]{Theorem}
\newtheorem*{theorem*}{}
 
\newtheorem{lemma}[theo]{Lemma}

\newtheorem{defn}[theo]{Definition}
\newtheorem{remk}[theo]{Remark}

\newcommand{\id}{\mathbb{I}}
\newcommand{\cH}{\mathcal{H}}
\newcommand{\Tr}{\mathrm{tr}}
\newcommand{\aQ}{\widetilde{\mathcal{Q}}}

\usepackage{rotating}
\usepackage{floatpag}
\usepackage{comment}
\rotfloatpagestyle{empty}

\usepackage{hyperref}

\usepackage{amsmath,amsthm,amssymb}
\usepackage{xspace,enumerate,color,epsfig}
\usepackage{graphicx}
\graphicspath{{.}{./figures/}}
\usepackage{marvosym}

\usepackage{tikzfig}
\usepackage{stmaryrd}
\usepackage{docmute}
\usepackage{keycommand}

\usepackage{pifont}

\usepackage{enumitem}

\newkeycommand{\pointmap}[style=point][1]{\,\tikz{\node[style=\commandkey{style}] (x) at (0,-0.3) {$#1$};\node [style=none] (3) at (0, 1.0) {};\node [style=none] (3) at (0, -1.0) {};\draw (x)--(0,0.7);}\,}

\newkeycommand{\typedkpoint}[style=kpoint,edgestyle=][2]{%
\,\begin{tikzpicture}
  \begin{pgfonlayer}{nodelayer}
    \node [style=none] (0) at (0, 1) {};
    \node [style=\commandkey{style}] (1) at (0, -0.75) {$#2$};
    \node [style=right label] (2) at (0.25, 0.5) {$#1$};
  \end{pgfonlayer}
  \begin{pgfonlayer}{edgelayer}
    \draw [\commandkey{edgestyle}] (1) to (0.center);
  \end{pgfonlayer}
\end{tikzpicture}\,}

\newkeycommand{\typedkpointdag}[style=kpoint adjoint,edgestyle=][2]{%
\,\begin{tikzpicture}
  \begin{pgfonlayer}{nodelayer}
    \node [style=none] (0) at (0, -1) {};
    \node [style=\commandkey{style}] (1) at (0, 0.75) {$#2$};
    \node [style=right label] (2) at (0.25, -0.5) {$#1$};
  \end{pgfonlayer}
  \begin{pgfonlayer}{edgelayer}
    \draw [\commandkey{edgestyle}] (0.center) to (1);
  \end{pgfonlayer}
\end{tikzpicture}\,}

\newcommand{\bistateadj}[1]{\,\begin{tikzpicture}[yshift=-3mm]
  \begin{pgfonlayer}{nodelayer}
    \node [style=kpointadj, minimum width=1 cm, inner sep=2pt] (0) at (0, 0.5) {$\psi$};
    \node [style=none] (1) at (-0.75, 0.25) {};
    \node [style=none] (2) at (0.75, 0.25) {};
    \node [style=none] (3) at (-0.75, -0.5) {};
    \node [style=none] (4) at (0.75, -0.5) {};
  \end{pgfonlayer}
  \begin{pgfonlayer}{edgelayer}
    \draw (1.center) to (3.center);
    \draw (2.center) to (4.center);
  \end{pgfonlayer}
\end{tikzpicture}\,}

\newkeycommand{\pointketbra}[style1=copoint,style2=point][2]{\,%
\begin{tikzpicture}
  \begin{pgfonlayer}{nodelayer}
    \node [style=none] (0) at (0, -1.5) {};
    \node [style=\commandkey{style1}] (1) at (0, -0.75) {$#1$};
    \node [style=\commandkey{style2}] (2) at (0, 0.75) {$#2$};
    \node [style=none] (3) at (0, 1.5) {};
  \end{pgfonlayer}
  \begin{pgfonlayer}{edgelayer}
    \draw (0.center) to (1);
    \draw (2) to (3.center);
  \end{pgfonlayer}
\end{tikzpicture}\,}

\newkeycommand{\twocopointketbra}[style1=copoint,style2=copoint,style3=point][3]{\,%
\begin{tikzpicture}
  \begin{pgfonlayer}{nodelayer}
    \node [style=none] (0) at (-0.75, -1.5) {};
    \node [style=none] (1) at (0.75, -1.5) {};
    \node [style=\commandkey{style1}] (2) at (-0.75, -0.75) {$#1$};
    \node [style=\commandkey{style2}] (3) at (0.75, -0.75) {$#2$};
    \node [style=\commandkey{style3}] (4) at (0, 0.75) {$#3$};
    \node [style=none] (5) at (0, 1.5) {};
  \end{pgfonlayer}
  \begin{pgfonlayer}{edgelayer}
    \draw (0.center) to (2);
    \draw (1.center) to (3);
    \draw (4) to (5.center);
  \end{pgfonlayer}
\end{tikzpicture}\,}

\newkeycommand{\twopointketbra}[style1=copoint,style2=point,style3=point][3]{\,%
\begin{tikzpicture}
  \begin{pgfonlayer}{nodelayer}
    \node [style=none] (0) at (0, -1.5) {};
    \node [style=none] (1) at (-0.75, 1.5) {};
    \node [style=\commandkey{style1}] (2) at (0, -0.75) {$#1$};
    \node [style=\commandkey{style2}] (3) at (-0.75, 0.75) {$#2$};
    \node [style=\commandkey{style3}] (4) at (0.75, 0.75) {$#3$};
    \node [style=none] (5) at (0.75, 1.5) {};
  \end{pgfonlayer}
  \begin{pgfonlayer}{edgelayer}
    \draw (0.center) to (2);
    \draw (1.center) to (3);
    \draw (4) to (5.center);
  \end{pgfonlayer}
\end{tikzpicture}\,}

\newkeycommand{\pointbraket}[style1=point,style2=copoint][2]{\,%
\begin{tikzpicture}
  \begin{pgfonlayer}{nodelayer}
    \node [style=\commandkey{style1}] (0) at (0, -0.5) {$#1$};
    \node [style=\commandkey{style2}] (1) at (0, 0.5) {$#2$};
  \end{pgfonlayer}
  \begin{pgfonlayer}{edgelayer}
    \draw (0) to (1);
  \end{pgfonlayer}
\end{tikzpicture}\,}

\newkeycommand{\kpointbraket}[style1=kpoint,style2=kpoint adjoint][2]{\,%
\begin{tikzpicture}
  \begin{pgfonlayer}{nodelayer}
    \node [style=\commandkey{style1}] (0) at (0, -0.75) {$#1$};
    \node [style=\commandkey{style2}] (1) at (0, 0.75) {$#2$};
  \end{pgfonlayer}
  \begin{pgfonlayer}{edgelayer}
    \draw (0) to (1);
  \end{pgfonlayer}
\end{tikzpicture}\,}

\newkeycommand{\kpointmap}[style1=kpoint,style2=map][2]{\,%
\begin{tikzpicture}
  \begin{pgfonlayer}{nodelayer}
    \node [style=\commandkey{style1}] (0) at (0, -1.1) {$#1$};
    \node [style=\commandkey{style2}] (1) at (0, 0.2) {$#2$};
    \node [style=none] (2) at (0, 1.5) {};
  \end{pgfonlayer}
  \begin{pgfonlayer}{edgelayer}
    \draw (0) to (1);
    \draw (1) to (2.center);
  \end{pgfonlayer}
\end{tikzpicture}%
\,}

\newkeycommand{\sandwichtwo}[style1=point,style2=map,style3=map,style4=copoint][4]{\,%
\begin{tikzpicture}
  \begin{pgfonlayer}{nodelayer}
    \node [style=\commandkey{style2}] (0) at (0, -0.75) {$#2$};
    \node [style=\commandkey{style3}] (1) at (0, 0.75) {$#3$};
    \node [style=\commandkey{style1}] (2) at (0, -2) {$#1$};
    \node [style=\commandkey{style4}] (3) at (0, 2) {$#4$};
  \end{pgfonlayer}
  \begin{pgfonlayer}{edgelayer}
    \draw (2) to (0);
    \draw (0) to (1);
    \draw (1) to (3);
  \end{pgfonlayer}
\end{tikzpicture}\,}

\newkeycommand{\longbraket}[style1=point,style2=copoint][2]{\,%
\begin{tikzpicture}
  \begin{pgfonlayer}{nodelayer}
    \node [style=\commandkey{style1}] (0) at (0, -2) {$#1$};
    \node [style=\commandkey{style2}] (1) at (0, 2) {$#2$};
  \end{pgfonlayer}
  \begin{pgfonlayer}{edgelayer}
    \draw (0) to (1);
  \end{pgfonlayer}
\end{tikzpicture}\,}

\newkeycommand{\onb}[style=point]{\ensuremath{\left\{\tikz{\node[style=\commandkey{style}] (x) at (0,-0.3) {$j$};\draw (x)--(0,0.7);}\right\}}\xspace}

\newkeycommand{\copointmap}[style=copoint][1]{\,\tikz{\node[style=\commandkey{style}] (x) at (0,0.3) {$#1$};\draw (x)--(0,-0.7);}\,}

\tikzstyle{cdiag}=[matrix of math nodes, row sep=3em, column sep=3em, text height=1.5ex, text depth=0.25ex,inner sep=0.5em]
\tikzstyle{arrow above}=[transform canvas={yshift=0.5ex}]
\tikzstyle{arrow below}=[transform canvas={yshift=-0.5ex}]

\usetikzlibrary{shapes.multipart}

\tikzset{->-/.style={decoration={
  markings,
  mark=at position .5 with {\arrow{>}}},postaction={decorate}}}
\tikzset{-<-/.style={decoration={
  markings,
  mark=at position .5 with {\arrow{<}}},postaction={decorate}}}

\tikzstyle{bwSpider}=[
       rectangle split,
       rectangle split parts=2,
       rectangle split part fill={black,white},
 minimum size=3.6 mm, inner sep=-2mm, draw=black,scale=0.5,rounded corners=0.8 mm
       ]
 \tikzstyle{wbSpider}=[
       rectangle split,
       rectangle split parts=2,
       rectangle split part fill={white,black},
 minimum size=3.6 mm, inner sep=-2mm, draw=black,scale=0.5,rounded corners=0.8 mm
       ]

\tikzstyle{epiCopoint}=[regular polygon,regular polygon sides=3,draw,scale=0.75,inner sep=-0.5pt,minimum width=5mm,fill=white,regular polygon rotate=0,line width=1pt]
\tikzstyle{epiPoint}=[regular polygon,regular polygon sides=3,draw,scale=0.75,inner sep=-0.5pt,minimum width=5mm,fill=white,regular polygon rotate=180,line width=1pt]
\tikzstyle{epiPointWide}=[regular polygon,regular polygon sides=3,draw,scale=0.75,inner sep=-0.5pt,minimum width=8mm,fill=white,regular polygon rotate=180,line width=1pt]
\tikzstyle{epiBox}=[fill=white,draw, line width = 1pt,inner sep=0.6mm,font=\footnotesize,minimum height=3mm,minimum width=3mm]
\tikzstyle{epiBoxWide}=[fill=white,draw, line width = 1pt,inner sep=0.6mm,font=\footnotesize,minimum height=3mm,minimum width=5mm]
\tikzstyle{epiBoxVeryWide}=[fill=white,draw, line width = 1pt,inner sep=0.6mm,font=\footnotesize,minimum height=3mm,minimum width=7mm]
\tikzstyle{qWire}=[line width = 1pt, color=black]
\tikzstyle{cWire}=[color=gray,line width = .75pt]
\tikzstyle{CqWire}=[color=gray,line width = .75pt,->-]
\tikzstyle{CcWire}=[color=gray,line width = .75pt,->-]
\tikzstyle{RqWire}=[line width = 1pt, color=black,-<-]
\tikzstyle{RcWire}=[color=gray,line width = .75pt,-<-]
\tikzstyle{env}=[copoint,regular polygon rotate=0,minimum width=0.2cm, fill=black]

\tikzstyle{probs}=[shape=semicircle,fill=white,draw=black,shape border rotate=180,minimum width=1.2cm]

\tikzstyle{every picture}=[baseline=-0.25em,scale=0.5]
\tikzstyle{dotpic}=[] 
\tikzstyle{diredges}=[every to/.style={diredge}]
\tikzstyle{math matrix}=[matrix of math nodes,left delimiter=(,right delimiter=),inner sep=2pt,column sep=1em,row sep=0.5em,nodes={inner sep=0pt},text height=1.5ex, text depth=0.25ex]

\tikzstyle{inline text}=[text height=1.5ex, text depth=0.25ex,yshift=0.5mm]
\tikzstyle{label}=[font=\footnotesize,text height=1.5ex, text depth=0.25ex,yshift=0.5mm]
\tikzstyle{left label}=[label,anchor=east,xshift=1.5mm]
\tikzstyle{right label}=[label,anchor=west,xshift=-1mm]
\tikzstyle{up label}=[label,anchor=south,yshift=-1mm]

\tikzstyle{braceedge}=[decorate,decoration={brace,amplitude=2mm,raise=-1mm}]
\tikzstyle{small braceedge}=[decorate,decoration={brace,amplitude=1mm,raise=-1mm}]

\tikzstyle{doubled}=[line width=1.6pt] 
\tikzstyle{boldedge}=[doubled,shorten <=-0.17mm,shorten >=-0.17mm]
\tikzstyle{boldedgegray}=[doubled,gray,shorten <=-0.17mm,shorten >=-0.17mm]
\tikzstyle{singleedgegray}=[gray]

\tikzstyle{semidoubled}=[line width=1.4pt] 
\tikzstyle{semiboldedgegray}=[semidoubled,gray,shorten <=-0.17mm,shorten >=-0.17mm]

\tikzstyle{boxedge}=[semiboldedgegray]

\tikzstyle{boldedgedashed}=[very thick,dashed,shorten <=-0.17mm,shorten >=-0.17mm]
\tikzstyle{vboldedgedashed}=[doubled,dashed,shorten <=-0.17mm,shorten >=-0.17mm]
\tikzstyle{left hook arrow}=[left hook-latex]
\tikzstyle{right hook arrow}=[right hook-latex]
\tikzstyle{sembracket}=[line width=0.5pt,shorten <=-0.07mm,shorten >=-0.07mm]

\tikzstyle{causal edge}=[->,thick,gray]
\tikzstyle{causal nondir}=[thick,gray]
\tikzstyle{timeline}=[thick,gray, dashed]

\tikzstyle{cedge}=[<->,thick,gray!70!white]

\tikzstyle{empty diagram}=[draw=gray!40!white,dashed,shape=rectangle,minimum width=1cm,minimum height=1cm]
\tikzstyle{empty diagram small}=[draw=gray!50!white,dashed,shape=rectangle,minimum width=0.6cm,minimum height=0.5cm]

\tikzstyle{dot}=[inner sep=0mm,minimum width=2mm,minimum height=2mm,draw,shape=circle]
\tikzstyle{bigdot}=[inner sep=0mm,minimum width=5mm,minimum height=5mm,draw,shape=circle]
\tikzstyle{leak}=[white dot, shape=regular polygon, minimum size=3.3 mm, regular polygon sides=3, outer sep=-0.2mm, regular polygon rotate=270]
\tikzstyle{proj}=[regular polygon,regular polygon sides=4,draw,scale=0.75,inner sep=-0.5pt,minimum width=6mm,fill=white]
\tikzstyle{projOut}=[regular polygon,regular polygon sides=3,draw,scale=0.75,inner sep=-0.5pt,minimum width=7.5mm,fill=white,regular polygon rotate=180]
\tikzstyle{projIn}=[regular polygon,regular polygon sides=3,draw,scale=0.75,inner sep=-0.5pt,minimum width=7.5mm,fill=white]
\tikzstyle{Vleak}=[white dot, shape=regular polygon, minimum size=3.3 mm, regular polygon sides=3, outer sep=-0.2mm, regular polygon rotate=90]
\tikzstyle{dleak}=[white dot, line width=1.6pt, shape=regular polygon, minimum size=3.3 mm, regular polygon sides=3, outer sep=-0.2mm, regular polygon rotate=270]

\tikzstyle{Wsquare}=[white dot, shape=regular polygon, rounded corners=0.8 mm, minimum size=3.3 mm, regular polygon sides=3, outer sep=-0.2mm]
\tikzstyle{Wsquareadj}=[white dot, shape=regular polygon, rounded corners=0.8 mm, minimum size=3.3 mm, regular polygon sides=3, outer sep=-0.2mm, regular polygon rotate=180]
\tikzstyle{ddot}=[inner sep=0mm, doubled, minimum width=2.5mm,minimum height=2.5mm,draw,shape=circle]

\tikzstyle{clear dot}=[dot,fill=none,text depth=-0.2mm,draw=gray, line width = .75pt]
\tikzstyle{tall clear dot}=[dot,fill=none,text depth=-0.2mm,draw=gray, line width = .75pt,shape=ellipse, minimum height=5mm]
\tikzstyle{wide clear dot}=[dot,fill=none,text depth=-0.2mm,draw=gray, line width = .75pt, shape=ellipse, minimum width = 5mm]
\tikzstyle{very wide clear dot}=[dot,fill=none,text depth=-0.2mm,draw=gray, line width = .75pt, shape=ellipse, minimum width = 7mm ]

\tikzstyle{black dot}=[dot,fill=black]
\tikzstyle{white dot}=[dot,fill=white,,text depth=-0.2mm]
\tikzstyle{white Wsquare}=[Wsquare,fill=gray,,text depth=-0.2mm]
\tikzstyle{white Wsquareadj}=[Wsquareadj,fill=white,,text depth=-0.2mm]
\tikzstyle{green dot}=[white dot] 
\tikzstyle{gray dot}=[dot,fill=gray!40!white,,text depth=-0.2mm]
\tikzstyle{red dot}=[gray dot] 

\tikzstyle{black ddot}=[ddot,fill=black]
\tikzstyle{white ddot}=[ddot,fill=white]
\tikzstyle{gray ddot}=[ddot,fill=gray!40!white]

\tikzstyle{gray edge}=[gray!60!white]

\tikzstyle{small dot}=[inner sep=0.2mm,minimum width=0pt,minimum height=0pt,draw,shape=circle]

\tikzstyle{small black dot}=[small dot,fill=black]
\tikzstyle{small white dot}=[small dot,fill=white]
\tikzstyle{small gray dot}=[small dot,fill=gray,draw=gray]

\tikzstyle{causal dot}=[inner sep=0.4mm,minimum width=0pt,minimum height=0pt,draw=white,shape=circle,fill=gray!40!white]

\tikzstyle{phase dimensions}=[minimum size=5mm,font=\footnotesize,rectangle,rounded corners=2.5mm,inner sep=0.2mm,outer sep=-2mm]
\tikzstyle{dphase dimensions}=[minimum size=5mm,font=\footnotesize,rectangle,rounded corners=2.5mm,inner sep=0.2mm,outer sep=-2mm]

\tikzstyle{white phase dot}=[dot,fill=white,phase dimensions]
\tikzstyle{white phase ddot}=[ddot,fill=white,dphase dimensions]

\tikzstyle{white rect ddot}=[draw=black,fill=white,doubled,minimum size=5mm,font=\footnotesize,rectangle,rounded corners=2.5mm,inner sep=0.2mm]
\tikzstyle{gray rect ddot}=[draw=black,fill=gray!40!white,doubled,minimum size=6mm,font=\footnotesize,rectangle,rounded corners=3mm]

\tikzstyle{gray phase dot}=[dot,fill=gray!40!white,phase dimensions]
\tikzstyle{gray phase ddot}=[ddot,fill=gray!40!white,dphase dimensions]
\tikzstyle{grey phase dot}=[gray phase dot]
\tikzstyle{grey phase ddot}=[gray phase ddot]

\tikzstyle{small phase dimensions}=[minimum size=4mm,font=\tiny,rectangle,rounded corners=2mm,inner sep=0.2mm,outer sep=-2mm]
\tikzstyle{small dphase dimensions}=[minimum size=4mm,font=\tiny,rectangle,rounded corners=2mm,inner sep=0.2mm,outer sep=-2mm]

\tikzstyle{small gray phase dot}=[dot,fill=gray!40!white,small phase dimensions]
\tikzstyle{small gray phase ddot}=[ddot,fill=gray!40!white,small dphase dimensions]

\tikzstyle{small map}=[draw,shape=rectangle,minimum height=4mm,minimum width=4mm,fill=white]

\tikzstyle{cnot}=[fill=white,shape=circle,inner sep=-1.4pt]

\tikzstyle{asym hadamard}=[fill=white,draw,shape=NEbox,inner sep=0.6mm,font=\footnotesize,minimum height=4mm]
\tikzstyle{asym hadamard conj}=[fill=white,draw,shape=NWbox,inner sep=0.6mm,font=\footnotesize,minimum height=4mm]
\tikzstyle{asym hadamard dag}=[fill=white,draw,shape=SEbox,inner sep=0.6mm,font=\footnotesize,minimum height=4mm]

\tikzstyle{hadamard}=[fill=white,draw,inner sep=0.6mm,font=\footnotesize,minimum height=4mm,minimum width=4mm]
\tikzstyle{small hadamard}=[fill=white,draw,inner sep=0.6mm,minimum height=1.5mm,minimum width=1.5mm]
\tikzstyle{small hadamard rotate}=[small hadamard,rotate=45]
\tikzstyle{dhadamard}=[hadamard,doubled]
\tikzstyle{small dhadamard}=[small hadamard,doubled]
\tikzstyle{small dhadamard rotate}=[small hadamard rotate,doubled]
\tikzstyle{antipode}=[white dot,inner sep=0.3mm,font=\footnotesize]

\tikzstyle{scalar}=[diamond,draw,inner sep=0.5pt,font=\small]
\tikzstyle{dscalar}=[diamond,doubled, draw,inner sep=0.5pt,font=\small]

\tikzstyle{small box}=[rectangle,inline text,fill=white,draw,minimum height=5mm,yshift=-0.5mm,minimum width=5mm,font=\small]
\tikzstyle{small gray box}=[small box,fill=gray!30]
\tikzstyle{medium box}=[rectangle,inline text,fill=white,draw,minimum height=5mm,yshift=-0.5mm,minimum width=10mm,font=\small]
\tikzstyle{square box}=[small box] 
\tikzstyle{medium gray box}=[small box,fill=gray!30]
\tikzstyle{semilarge box}=[rectangle,inline text,fill=white,draw,minimum height=5mm,yshift=-0.5mm,minimum width=12.5mm,font=\small]
\tikzstyle{large box}=[rectangle,inline text,fill=white,draw,minimum height=5mm,yshift=-0.5mm,minimum width=15mm,font=\small]
\tikzstyle{large gray box}=[small box,fill=gray!30]

\tikzstyle{Bayes box}=[rectangle,fill=black,draw, minimum height=3mm, minimum width=3mm]

\tikzstyle{gray square point}=[small box,fill=gray!50]

\tikzstyle{dphase box white}=[dhadamard]
\tikzstyle{dphase box gray}=[dhadamard,fill=gray!50!white]
\tikzstyle{phase box white}=[hadamard]
\tikzstyle{phase box gray}=[hadamard,fill=gray!50!white]

\tikzstyle{point}=[regular polygon,regular polygon sides=3,draw,scale=0.75,inner sep=-0.5pt,minimum width=9mm,fill=white,regular polygon rotate=180]
\tikzstyle{infpoint}=[regular polygon,regular polygon sides=3,draw,scale=0.75,inner sep=-0.5pt,minimum width=9mm,fill=white,regular polygon rotate=90]
\tikzstyle{point nosep}=[regular polygon,regular polygon sides=3,draw,scale=0.75,inner sep=-2pt,minimum width=9mm,fill=white,regular polygon rotate=180]
\tikzstyle{infcopoint}=[regular polygon,regular polygon sides=3,draw,scale=0.75,inner sep=-0.5pt,minimum width=9mm,fill=white,regular polygon rotate=270]
\tikzstyle{copoint}=[regular polygon,regular polygon sides=3,draw,scale=0.75,inner sep=-0.5pt,minimum width=9mm,fill=white]
\tikzstyle{dpoint}=[point,doubled]
\tikzstyle{dcopoint}=[copoint,doubled]

\tikzstyle{pointgrow}=[shape=cornerpoint,kpoint common,scale=0.75,inner sep=3pt]
\tikzstyle{pointgrow dag}=[shape=cornercopoint,kpoint common,scale=0.75,inner sep=3pt]

\tikzstyle{wide copoint}=[fill=white,draw,shape=isosceles triangle,shape border rotate=90,isosceles triangle stretches=true,inner sep=0pt,minimum width=1.5cm,minimum height=6.12mm]
\tikzstyle{wide point}=[fill=white,draw,shape=isosceles triangle,shape border rotate=-90,isosceles triangle stretches=true,inner sep=0pt,minimum width=1.5cm,minimum height=6.12mm,yshift=-0.0mm]
\tikzstyle{wide point plus}=[fill=white,draw,shape=isosceles triangle,shape border rotate=-90,isosceles triangle stretches=true,inner sep=0pt,minimum width=1.74cm,minimum height=7mm,yshift=-0.0mm]

\tikzstyle{wide dpoint}=[fill=white,doubled,draw,shape=isosceles triangle,shape border rotate=-90,isosceles triangle stretches=true,inner sep=0pt,minimum width=1.5cm,minimum height=6.12mm,yshift=-0.0mm]

\tikzstyle{tinypoint}=[regular polygon,regular polygon sides=3,draw,scale=0.55,inner sep=-0.15pt,minimum width=6mm,fill=white,regular polygon rotate=180]

\tikzstyle{white point}=[point]
\tikzstyle{white dpoint}=[dpoint]
\tikzstyle{green point}=[white point] 
\tikzstyle{white copoint}=[copoint]
\tikzstyle{gray point}=[point,fill=gray!40!white]
\tikzstyle{gray dpoint}=[gray point,doubled]
\tikzstyle{red point}=[gray point] 
\tikzstyle{gray copoint}=[copoint,fill=gray!40!white]
\tikzstyle{gray dcopoint}=[gray copoint,doubled]

\tikzstyle{white point guide}=[regular polygon,regular polygon sides=3,font=\scriptsize,draw,scale=0.65,inner sep=-0.5pt,minimum width=9mm,fill=white,regular polygon rotate=180]

\tikzstyle{black point}=[point,fill=black,font=\color{white}]
\tikzstyle{black copoint}=[copoint,fill=black,font=\color{white}]

\tikzstyle{tiny gray point}=[tinypoint,fill=gray!40!white]

\tikzstyle{diredge}=[->]
\tikzstyle{ddiredge}=[<->]
\tikzstyle{rdiredge}=[<-]
\tikzstyle{thickdiredge}=[->, very thick]
\tikzstyle{pointer edge}=[->,very thick,gray]
\tikzstyle{pointer edge part}=[very thick,gray]
\tikzstyle{dashed edge}=[dashed]
\tikzstyle{thick dashed edge}=[very thick,dashed]
\tikzstyle{thick gray dashed edge}=[thick dashed edge,gray!40]
\tikzstyle{thick map edge}=[very thick,|->]

\makeatletter
\newcommand{\boxshape}[3]{%
\pgfdeclareshape{#1}{
\inheritsavedanchors[from=rectangle] 
\inheritanchorborder[from=rectangle]
\inheritanchor[from=rectangle]{center}
\inheritanchor[from=rectangle]{north}
\inheritanchor[from=rectangle]{south}
\inheritanchor[from=rectangle]{west}
\inheritanchor[from=rectangle]{east}
\backgroundpath{
\southwest \pgf@xa=\pgf@x \pgf@ya=\pgf@y
\northeast \pgf@xb=\pgf@x \pgf@yb=\pgf@y

\@tempdima=#2
\@tempdimb=#3

\pgfpathmoveto{\pgfpoint{\pgf@xa - 5pt + \@tempdima}{\pgf@ya}}
\pgfpathlineto{\pgfpoint{\pgf@xa - 5pt - \@tempdima}{\pgf@yb}}
\pgfpathlineto{\pgfpoint{\pgf@xb + 5pt + \@tempdimb}{\pgf@yb}}
\pgfpathlineto{\pgfpoint{\pgf@xb + 5pt - \@tempdimb}{\pgf@ya}}
\pgfpathlineto{\pgfpoint{\pgf@xa - 5pt + \@tempdima}{\pgf@ya}}
\pgfpathclose
}
}}

\boxshape{NEbox}{0pt}{3pt}
\boxshape{SEbox}{0pt}{-3pt}
\boxshape{NWbox}{5pt}{0pt}
\boxshape{SWbox}{-5pt}{0pt}
\boxshape{EBox}{-3pt}{3pt}
\boxshape{WBox}{3pt}{-3pt}
\makeatother

\tikzstyle{cloud}=[shape=cloud,draw,minimum width=1.5cm,minimum height=1.5cm]

\tikzstyle{map}=[draw,shape=NEbox,inner sep=1pt,minimum height=4mm,fill=white]
\tikzstyle{dashedmap}=[draw,dashed,shape=NEbox,inner sep=2pt,minimum height=6mm,fill=white]
\tikzstyle{mapdag}=[draw,shape=SEbox,inner sep=1pt,minimum height=4mm,fill=white]
\tikzstyle{mapadj}=[draw,shape=SEbox,inner sep=2pt,minimum height=6mm,fill=white]
\tikzstyle{maptrans}=[draw,shape=SWbox,inner sep=2pt,minimum height=6mm,fill=white]
\tikzstyle{mapconj}=[draw,shape=NWbox,inner sep=2pt,minimum height=6mm,fill=white]

\tikzstyle{medium map}=[draw,shape=NEbox,inner sep=2pt,minimum height=6mm,fill=white,minimum width=7mm]
\tikzstyle{medium map dag}=[draw,shape=SEbox,inner sep=2pt,minimum height=6mm,fill=white,minimum width=7mm]
\tikzstyle{medium map adj}=[draw,shape=SEbox,inner sep=2pt,minimum height=6mm,fill=white,minimum width=7mm]
\tikzstyle{medium map trans}=[draw,shape=SWbox,inner sep=2pt,minimum height=6mm,fill=white,minimum width=7mm]
\tikzstyle{medium map conj}=[draw,shape=NWbox,inner sep=2pt,minimum height=6mm,fill=white,minimum width=7mm]
\tikzstyle{semilarge map}=[draw,shape=NEbox,inner sep=2pt,minimum height=6mm,fill=white,minimum width=9.5mm]
\tikzstyle{semilarge map trans}=[draw,shape=SWbox,inner sep=2pt,minimum height=6mm,fill=white,minimum width=9.5mm]
\tikzstyle{semilarge map adj}=[draw,shape=SEbox,inner sep=2pt,minimum height=6mm,fill=white,minimum width=9.5mm]
\tikzstyle{semilarge map dag}=[draw,shape=SEbox,inner sep=2pt,minimum height=6mm,fill=white,minimum width=9.5mm]
\tikzstyle{semilarge map conj}=[draw,shape=NWbox,inner sep=2pt,minimum height=6mm,fill=white,minimum width=9.5mm]
\tikzstyle{large map}=[draw,shape=NEbox,inner sep=2pt,minimum height=6mm,fill=white,minimum width=12mm]
\tikzstyle{large map conj}=[draw,shape=NWbox,inner sep=2pt,minimum height=6mm,fill=white,minimum width=12mm]
\tikzstyle{very large map}=[draw,shape=NEbox,inner sep=2pt,minimum height=6mm,fill=white,minimum width=17mm]

\tikzstyle{medium dmap}=[draw,doubled,shape=NEbox,inner sep=2pt,minimum height=6mm,fill=white,minimum width=7mm]
\tikzstyle{medium dmap dag}=[draw,doubled,shape=SEbox,inner sep=2pt,minimum height=6mm,fill=white,minimum width=7mm]
\tikzstyle{medium dmap adj}=[draw,doubled,shape=SEbox,inner sep=2pt,minimum height=6mm,fill=white,minimum width=7mm]
\tikzstyle{medium dmap trans}=[draw,doubled,shape=SWbox,inner sep=2pt,minimum height=6mm,fill=white,minimum width=7mm]
\tikzstyle{medium dmap conj}=[draw,doubled,shape=NWbox,inner sep=2pt,minimum height=6mm,fill=white,minimum width=7mm]
\tikzstyle{semilarge dmap}=[draw,doubled,shape=NEbox,inner sep=2pt,minimum height=6mm,fill=white,minimum width=9.5mm]
\tikzstyle{semilarge dmap trans}=[draw,doubled,shape=SWbox,inner sep=2pt,minimum height=6mm,fill=white,minimum width=9.5mm]
\tikzstyle{semilarge dmap adj}=[draw,doubled,shape=SEbox,inner sep=2pt,minimum height=6mm,fill=white,minimum width=9.5mm]
\tikzstyle{semilarge dmap dag}=[draw,doubled,shape=SEbox,inner sep=2pt,minimum height=6mm,fill=white,minimum width=9.5mm]
\tikzstyle{semilarge dmap conj}=[draw,doubled,shape=NWbox,inner sep=2pt,minimum height=6mm,fill=white,minimum width=9.5mm]
\tikzstyle{large dmap}=[draw,doubled,shape=NEbox,inner sep=2pt,minimum height=6mm,fill=white,minimum width=12mm]
\tikzstyle{large dmap conj}=[draw,doubled,shape=NWbox,inner sep=2pt,minimum height=6mm,fill=white,minimum width=12mm]
\tikzstyle{large dmap trans}=[draw,doubled,shape=SWbox,inner sep=2pt,minimum height=6mm,fill=white,minimum width=12mm]
\tikzstyle{large dmap adj}=[draw,doubled,shape=SEbox,inner sep=2pt,minimum height=6mm,fill=white,minimum width=12mm]
\tikzstyle{large dmap dag}=[draw,doubled,shape=SEbox,inner sep=2pt,minimum height=6mm,fill=white,minimum width=12mm]
\tikzstyle{very large dmap}=[draw,doubled,shape=NEbox,inner sep=2pt,minimum height=6mm,fill=white,minimum width=19.5mm]

\tikzstyle{muxbox}=[draw,shape=rectangle,minimum height=3mm,minimum width=3mm,fill=white]
\tikzstyle{dmuxbox}=[muxbox,doubled]

\tikzstyle{box}=[draw,shape=rectangle,inner sep=2pt,minimum height=6mm,minimum width=6mm,fill=white]
\tikzstyle{dbox}=[draw,doubled,shape=rectangle,inner sep=2pt,minimum height=6mm,minimum width=6mm,fill=white]
\tikzstyle{dmap}=[draw,doubled,shape=NEbox,inner sep=2pt,minimum height=6mm,fill=white]
\tikzstyle{dmapdag}=[draw,doubled,shape=SEbox,inner sep=2pt,minimum height=6mm,fill=white]
\tikzstyle{dmapadj}=[draw,doubled,shape=SEbox,inner sep=2pt,minimum height=6mm,fill=white]
\tikzstyle{dmaptrans}=[draw,doubled,shape=SWbox,inner sep=2pt,minimum height=6mm,fill=white]
\tikzstyle{dmapconj}=[draw,doubled,shape=NWbox,inner sep=2pt,minimum height=6mm,fill=white]

\tikzstyle{ddmap}=[draw,doubled,dashed,shape=NEbox,inner sep=2pt,minimum height=6mm,fill=white]
\tikzstyle{ddmapdag}=[draw,doubled,dashed,shape=SEbox,inner sep=2pt,minimum height=6mm,fill=white]
\tikzstyle{ddmapadj}=[draw,doubled,dashed,shape=SEbox,inner sep=2pt,minimum height=6mm,fill=white]
\tikzstyle{ddmaptrans}=[draw,doubled,dashed,shape=SWbox,inner sep=2pt,minimum height=6mm,fill=white]
\tikzstyle{ddmapconj}=[draw,doubled,dashed,shape=NWbox,inner sep=2pt,minimum height=6mm,fill=white]

\boxshape{sNEbox}{0pt}{3pt}
\boxshape{sSEbox}{0pt}{-3pt}
\boxshape{sNWbox}{3pt}{0pt}
\boxshape{sSWbox}{-3pt}{0pt}
\tikzstyle{smap}=[draw,shape=sNEbox,fill=white]
\tikzstyle{smapdag}=[draw,shape=sSEbox,fill=white]
\tikzstyle{smapadj}=[draw,shape=sSEbox,fill=white]
\tikzstyle{smaptrans}=[draw,shape=sSWbox,fill=white]
\tikzstyle{smapconj}=[draw,shape=sNWbox,fill=white]

\tikzstyle{dsmap}=[draw,dashed,shape=sNEbox,fill=white]
\tikzstyle{dsmapdag}=[draw,dashed,shape=sSEbox,fill=white]
\tikzstyle{dsmaptrans}=[draw,dashed,shape=sSWbox,fill=white]
\tikzstyle{dsmapconj}=[draw,dashed,shape=sNWbox,fill=white]

\boxshape{mNEbox}{0pt}{10pt}
\boxshape{mSEbox}{0pt}{-10pt}
\boxshape{mNWbox}{10pt}{0pt}
\boxshape{mSWbox}{-10pt}{0pt}
\tikzstyle{mmap}=[draw,shape=mNEbox]
\tikzstyle{mmapdag}=[draw,shape=mSEbox]
\tikzstyle{mmaptrans}=[draw,shape=mSWbox]
\tikzstyle{mmapconj}=[draw,shape=mNWbox]

\tikzstyle{mmapgray}=[draw,fill=gray!40!white,shape=mNEbox]
\tikzstyle{smapgray}=[draw,fill=gray!40!white,shape=sNEbox]

\makeatletter

\pgfdeclareshape{cornerpoint}{
\inheritsavedanchors[from=rectangle] 
\inheritanchorborder[from=rectangle]
\inheritanchor[from=rectangle]{center}
\inheritanchor[from=rectangle]{north}
\inheritanchor[from=rectangle]{south}
\inheritanchor[from=rectangle]{west}
\inheritanchor[from=rectangle]{east}
\backgroundpath{
\southwest \pgf@xa=\pgf@x \pgf@ya=\pgf@y
\northeast \pgf@xb=\pgf@x \pgf@yb=\pgf@y

\pgfmathsetmacro{\pgf@shorten@left}{\pgfkeysvalueof{/tikz/shorten left}}
\pgfmathsetmacro{\pgf@shorten@right}{\pgfkeysvalueof{/tikz/shorten right}}

\pgfpathmoveto{\pgfpoint{0.5 * (\pgf@xa + \pgf@xb)}{\pgf@ya - 5pt}}
\pgfpathlineto{\pgfpoint{\pgf@xa - 8pt + \pgf@shorten@left}{\pgf@yb - 1.5 * \pgf@shorten@left}}
\pgfpathlineto{\pgfpoint{\pgf@xa - 8pt + \pgf@shorten@left}{\pgf@yb}}
\pgfpathlineto{\pgfpoint{\pgf@xb + 8pt - \pgf@shorten@right}{\pgf@yb}}
\pgfpathlineto{\pgfpoint{\pgf@xb + 8pt - \pgf@shorten@right}{\pgf@yb - 1.5 * \pgf@shorten@right}}
\pgfpathclose
}
}

\pgfdeclareshape{cornercopoint}{
\inheritsavedanchors[from=rectangle] 
\inheritanchorborder[from=rectangle]
\inheritanchor[from=rectangle]{center}
\inheritanchor[from=rectangle]{north}
\inheritanchor[from=rectangle]{south}
\inheritanchor[from=rectangle]{west}
\inheritanchor[from=rectangle]{east}
\backgroundpath{
\southwest \pgf@xa=\pgf@x \pgf@ya=\pgf@y
\northeast \pgf@xb=\pgf@x \pgf@yb=\pgf@y

\pgfmathsetmacro{\pgf@shorten@left}{\pgfkeysvalueof{/tikz/shorten left}}
\pgfmathsetmacro{\pgf@shorten@right}{\pgfkeysvalueof{/tikz/shorten right}}

\pgfpathmoveto{\pgfpoint{0.5 * (\pgf@xa + \pgf@xb)}{\pgf@yb + 5pt}}
\pgfpathlineto{\pgfpoint{\pgf@xa - 8pt + \pgf@shorten@left}{\pgf@ya + 1.5 * \pgf@shorten@left}}
\pgfpathlineto{\pgfpoint{\pgf@xa - 8pt + \pgf@shorten@left}{\pgf@ya}}
\pgfpathlineto{\pgfpoint{\pgf@xb + 8pt - \pgf@shorten@right}{\pgf@ya}}
\pgfpathlineto{\pgfpoint{\pgf@xb + 8pt - \pgf@shorten@right}{\pgf@ya + 1.5 * \pgf@shorten@right}}
\pgfpathclose
}
}

\makeatother

\pgfkeyssetvalue{/tikz/shorten left}{0pt}
\pgfkeyssetvalue{/tikz/shorten right}{0pt}

\tikzstyle{kpoint common}=[draw,fill=white,inner sep=1pt,minimum height=4mm]
\tikzstyle{kpoint sc}=[shape=cornerpoint,kpoint common]
\tikzstyle{kpoint adjoint sc}=[shape=cornercopoint,kpoint common]
\tikzstyle{kpoint}=[shape=cornerpoint,shorten left=5pt,kpoint common]
\tikzstyle{kpoint adjoint}=[shape=cornercopoint,shorten left=5pt,kpoint common]
\tikzstyle{kpoint conjugate}=[shape=cornerpoint,shorten right=5pt,kpoint common]
\tikzstyle{kpoint transpose}=[shape=cornercopoint,shorten right=5pt,kpoint common]
\tikzstyle{kpoint symm}=[shape=cornerpoint,shorten left=5pt,shorten right=5pt,kpoint common]

\tikzstyle{wide kpoint sc}=[shape=cornerpoint,kpoint common, minimum width=1 cm]
\tikzstyle{wide kpointdag sc}=[shape=cornercopoint,kpoint common, minimum width=1 cm]

\tikzstyle{black kpoint}=[shape=cornerpoint,shorten left=5pt,kpoint common,fill=black,font=\color{white}]

\tikzstyle{black kpoint sm}=[shape=cornerpoint,shorten left=5pt,kpoint common,fill=black,font=\color{white},scale=0.75]

\tikzstyle{black kpoint adjoint}=[shape=cornercopoint,shorten left=5pt,kpoint common,fill=black,font=\color{white}]
\tikzstyle{black kpointadj}=[shape=cornercopoint,shorten left=5pt,kpoint common,fill=black,font=\color{white}]

\tikzstyle{black kpointadj sm}=[shape=cornercopoint,shorten left=5pt,kpoint common,fill=black,font=\color{white},scale=0.75]

\tikzstyle{black dkpoint}=[shape=cornerpoint,shorten left=5pt,kpoint common,fill=black, doubled,font=\color{white}]
\tikzstyle{black dkpoint adjoint}=[shape=cornercopoint,shorten left=5pt,kpoint common,fill=black, doubled,font=\color{white}]
\tikzstyle{black dkpointadj}=[shape=cornercopoint,shorten left=5pt,kpoint common,fill=black, doubled,font=\color{white}]

\tikzstyle{black dkpoint sm}=[shape=cornerpoint,shorten left=5pt,kpoint common,fill=black, doubled,font=\color{white},scale=0.75]
\tikzstyle{black dkpointadj sm}=[shape=cornercopoint,shorten left=5pt,kpoint common,fill=black, doubled,font=\color{white},scale=0.75]

\tikzstyle{kpointdag}=[kpoint adjoint]
\tikzstyle{kpointadj}=[kpoint adjoint]
\tikzstyle{kpointconj}=[kpoint conjugate]
\tikzstyle{kpointtrans}=[kpoint transpose]

\tikzstyle{big kpoint}=[kpoint, minimum width=1.2 cm, minimum height=8mm, inner sep=4pt, text depth=3mm]

\tikzstyle{wide kpoint}=[kpoint, minimum width=1 cm, inner sep=2pt]
\tikzstyle{wide kpointdag}=[kpointdag, minimum width=1 cm, inner sep=2pt]
\tikzstyle{wide kpointconj}=[kpointconj, minimum width=1 cm, inner sep=2pt]
\tikzstyle{wide kpointtrans}=[kpointtrans, minimum width=1 cm, inner sep=2pt]

\tikzstyle{wider kpoint}=[kpoint, minimum width=1.25 cm, inner sep=2pt]
\tikzstyle{wider kpointdag}=[kpointdag, minimum width=1.25 cm, inner sep=2pt]
\tikzstyle{wider kpointconj}=[kpointconj, minimum width=1.25 cm, inner sep=2pt]
\tikzstyle{wider kpointtrans}=[kpointtrans, minimum width=1.25 cm, inner sep=2pt]

\tikzstyle{gray kpoint}=[kpoint,fill=gray!50!white]
\tikzstyle{gray kpointdag}=[kpointdag,fill=gray!50!white]
\tikzstyle{gray kpointadj}=[kpointadj,fill=gray!50!white]
\tikzstyle{gray kpointconj}=[kpointconj,fill=gray!50!white]
\tikzstyle{gray kpointtrans}=[kpointtrans,fill=gray!50!white]

\tikzstyle{gray dkpoint}=[kpoint,fill=gray!50!white,doubled]
\tikzstyle{gray dkpointdag}=[kpointdag,fill=gray!50!white,doubled]
\tikzstyle{gray dkpointadj}=[kpointadj,fill=gray!50!white,doubled]
\tikzstyle{gray dkpointconj}=[kpointconj,fill=gray!50!white,doubled]
\tikzstyle{gray dkpointtrans}=[kpointtrans,fill=gray!50!white,doubled]

\tikzstyle{white label}=[draw,fill=white,rectangle,inner sep=0.7 mm]
\tikzstyle{gray label}=[draw,fill=gray!50!white,rectangle,inner sep=0.7 mm]
\tikzstyle{black label}=[draw,fill=black,rectangle,inner sep=0.7 mm]

\tikzstyle{dkpoint}=[kpoint,doubled]
\tikzstyle{wide dkpoint}=[wide kpoint,doubled]
\tikzstyle{dkpointdag}=[kpoint adjoint,doubled]
\tikzstyle{wide dkpointdag}=[wide kpointdag,doubled]
\tikzstyle{dkcopoint}=[kpoint adjoint,doubled]
\tikzstyle{dkpointadj}=[kpoint adjoint,doubled]
\tikzstyle{dkpointconj}=[kpoint conjugate,doubled]
\tikzstyle{dkpointtrans}=[kpoint transpose,doubled]

\tikzstyle{kscalar}=[kpoint common, shape=EBox, inner xsep=-1pt, inner ysep=3pt,font=\small]
\tikzstyle{kscalarconj}=[kpoint common, shape=WBox, inner xsep=-1pt, inner ysep=3pt,font=\small]

\tikzstyle{spekpoint}=[kpoint sc,minimum height=5mm,inner sep=3pt]
\tikzstyle{spekcopoint}=[kpoint adjoint sc,minimum height=5mm,inner sep=3pt]

\tikzstyle{dspekpoint}=[spekpoint,doubled]
\tikzstyle{dspekcopoint}=[spekcopoint,doubled]

 \tikzstyle{upground}=[circuit ee IEC,thick,ground,rotate=90,scale=2.5]
 \tikzstyle{downground}=[circuit ee IEC,thick,ground,rotate=-90,scale=2.5]
 \tikzstyle{infupground}=[circuit ee IEC,thick,ground,rotate=0,scale=2.5]
 \tikzstyle{infdownground}=[circuit ee IEC,thick,ground,rotate=180,scale=2.5]
 \tikzstyle{bigground}=[regular polygon,regular polygon sides=3,draw=gray,scale=0.50,inner sep=-0.5pt,minimum width=10mm,fill=gray]

\tikzstyle{arrs}=[-latex,font=\small,auto]
\tikzstyle{arrow plain}=[arrs]
\tikzstyle{arrow dashed}=[dashed,arrs]
\tikzstyle{arrow bold}=[very thick,arrs]
\tikzstyle{arrow hide}=[draw=white!0,-]
\tikzstyle{arrow reverse}=[latex-]
\tikzstyle{cdnode}=[]

\let\olddagger\dagger
\renewcommand{\dagger}{\ensuremath{\olddagger}\xspace}

\usepackage{makeidx}
\makeindex

\usepackage{color}
\def\bR{\begin{color}{red}}
\def\bB{\begin{color}{blue}}
\def\bM{\begin{color}{magenta}}
\def\bC{\begin{color}{cyan}}
\def\bW{\begin{color}{white}}
\def\bBl{\begin{color}{black}}
\def\bG{\begin{color}{green}}
\def\bY{\begin{color}{yellow}}
\def\e{\end{color}\xspace}
\newcommand{\bit}{\begin{itemize}}
\newcommand{\eit}{\end{itemize}\par\noindent}
\newcommand{\ben}{\begin{enumerate}}
\newcommand{\een}{\end{enumerate}\par\noindent}
\newcommand{\beq}{\begin{equation}}
\newcommand{\eeq}{\end{equation}\par\noindent}
\newcommand{\beqa}{\begin{eqnarray*}}
\newcommand{\eeqa}{\end{eqnarray*}\par\noindent}
\newcommand{\beqn}{\begin{eqnarray}}
\newcommand{\eeqn}{\end{eqnarray}\par\noindent}

\def\jR{\begin{color}{black}}
\def\jB{\begin{color}{black}}
\def\jM{\begin{color}{magenta}}
\def\jC{\begin{color}{cyan}}
\def\jW{\begin{color}{white}}
\def\jBl{\begin{color}{black}}
\def\jG{\begin{color}{green}}
\def\jY{\begin{color}{yellow}}

\begin{document}
\title{Bipartite post-quantum steering in generalised scenarios}

\author{Ana Bel{\'e}n Sainz}
\affiliation{International Centre for Theory of Quantum Technologies, University of Gda\'nsk, 80-308 Gda\'nsk, Poland}
\affiliation{Perimeter Institute for Theoretical Physics, 31 Caroline St. N, Waterloo, Ontario, Canada, N2L 2Y5}
\author{{Matty J.~Hoban}}
\affiliation{Department of Computing, Goldsmiths, University of London, London SE14 6NW, UK}
\author{{Paul Skrzypczyk}}
\affiliation{H. H. Wills Physics Laboratory, University of Bristol, Tyndall Avenue, Bristol, BS8 1TL, UK}
\author{{Leandro Aolita}}
\affiliation{Instituto  de  F\'isica,  Universidade  Federal  do  Rio  de  Janeiro$\text{,}$ Caixa  Postal  68528,  21941-972  Rio  de  Janeiro,  RJ,  Brazil}

\date{\today}

\begin{abstract}
The study of stronger-than-quantum effects is a fruitful line of research that provides valuable insight into quantum theory. Unfortunately, traditional bipartite steering scenarios can always be explained by quantum theory. Here we show that, by relaxing this traditional setup, bipartite steering incompatible with quantum theory is possible. The two scenarios we describe, which still feature Alice remotely steering Bob's system, are: (i) one where Bob also has an input and operates on his subsystem, and (ii) the `instrumental steering' scenario. We show that such bipartite post-quantum steering is a genuinely new type of post-quantum nonlocality, which does not follow from post-quantum Bell nonlocality. {In addition, we present a method to bound quantum violations of steering inequalities in these scenarios.}
\end{abstract}

\maketitle

\textit{Introduction.--} 
Einstein-Podolsky-Rosen steering is a striking nonlocal feature of quantum theory \cite{sch, wise}. First discussed by Schr\"odinger \cite{sch}, it refers to the phenomenon where Alice, by performing measurements on half of a shared system, remotely `steers' the state of a distant Bob, in a way which has no classical explanation. From a modern quantum information perspective \cite{wise} steering certifies entanglement in situations where Alice's devices are  uncharacterised or untrusted, allowing for ``one-sided device independent'' implementations of information-theoretic tasks, such as quantum key distribution \cite{qkd}, randomness certification \cite{rand1, rand2}, measurement incompatibility certification \cite{meas1, meas2, meas3}, and self-testing \cite{self1, self2}.

{Given the usefulness of quantum steering as a resource for information processing, a comprehensive understanding of this non-classical phenomenon as a resource is highly desirable. A fruitful way to approach this, pursued in the study of other non-classical phenomena, e.g.~Bell nonlocality \cite{Bell} and contextuality \cite{KS}, is to investigate it `from the outside': namely, to study it operationally from the perspective of a more general theory --  which may supersede quantum theory --  and then understand which aspects are purely quantum. Studying phenomena beyond what quantum theory predicts is relevant not only from the hypothetical perspective of a post-quantum theory, but also -- and above all -- because it allows for a deeper understanding of the foundations of quantum theory and the limitations it has for information processing \cite{thenewref}. The main question studied here is how to properly understand steering from this more general perspective, on which we report substantial progress in this paper.}

Abstractly, we may view the steering scenario as one where Alice has a device that accepts a classical input, $x$, usually thought of as labelling the choice of measurement, and produces a classical outcome, $a$, usually thought of as the measurement result, while Bob has a device without an input, that produces a quantum system, which is correlated with the input and outcome of Alice, and usually thought of as the steered system. {Here we are interested in the possibility that the local structure of quantum theory is maintained, while considering more general global structure -- for instance more general types of correlations or global dynamics. }

{In this setting, we would like to re-examine the phenomenon of steering. A natural question that arises is whether a more general theory may allow for steering beyond what quantum theory predicts. That is, could it be possible to find a pair of devices for Alice and Bob which could not be produced within quantum theory, by Alice and Bob sharing a quantum state, upon which Alice performs measurements labelled by $x$ and with outcomes $a$? The only requirement that we maintain in this generalised setting is that of relativistic causality: Alice should not be able to use steering to signal to Bob, i.e., to send information to him instantaneously.  }

A celebrated theorem by Gisin \cite{gisin} and Hughston, Josza and Wootters \cite{HJW} (GHJW) shows that post-quantum steering cannot occur in the traditional setting. Namely, any pair of devices that do not allow signalling from Alice to Bob can always be realised by some carefully chosen set of measurements and quantum state. 
The traditional setting is however not the only interesting scenario where one can see the steering phenomena. 
In Ref.~\cite{pqs}, post-quantum \emph{multi-partite} steering was discovered: in a tri-partite scenario, Alice and Bob are able to jointly steer the state of a third party, Charlie in a way which cannot arise from measurements an any quantum state. 
Subsequently, unified frameworks for studying quantum and post-quantum steering in the multipartite setting have been developed, providing a playground for exploring this fascinating effect \cite{oform, chan}. 

A key question that nevertheless remained unanswered is whether it is possible to have post-quantum steering in a suitable generalised bipartite scenario, or whether post-quantum steering is a purely multipartite phenomenon. In this work we answer this question in the positive. We discover two natural bipartite generalisations of steering that allow for post-quantum effects (see Fig.~\ref{fig:steebip} cases (c) and (d) respectively): one where Bob also has an input that allows him to additionally influence his quantum state, and another where this additional influence is instead conditioned on Alice's outcome. This second generalisation corresponds to a specific type of setup known as the `instrumental causal network', that is ubiquitous in causal inference \cite{SGS01,Pearl09}. 
Furthermore, we show, crucially, that in both cases the post-quantum steering uncovered genuinely constitute new effects, that are distinct from post-quantum nonlocality in the associated generalised setups. We do this by finding explicit examples of post-quantum steering where if Bob {performs measurements on his quantum system, then the resulting outcome statistics} are never post-quantum nonlocal. {To prove our results, we develop a numerical technique to bound the quantum violations of steering inequalities.}

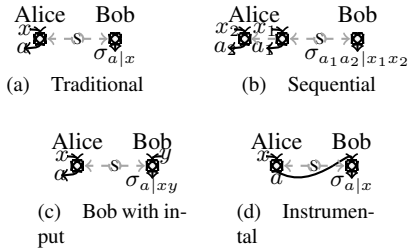
\begin{figure}
\begin{center}
\subfigure[\quad Traditional]{
		\begin{tikzpicture}[scale=0.5]
		\node at (-2,1.3) {Alice};
		\shade[draw, thick, ,rounded corners, inner color=white,outer color=gray!50!white] (-1.7,-0.3) rectangle (-2.3,0.3) ;
		\draw[thick, ->] (-2.5,0.5) to [out=180, in=90] (-2,0.3);
		\node at (-2.8,0.5) {$x$};
		\draw[thick, -<] (-2,-0.3) to [out=-90, in=180] (-2.5,-0.5);
		\node at (-2.9,-0.5) {$a$};

		\node at (2,1.3) {Bob};
		\shade[draw, thick, ,rounded corners, inner color=white,outer color=gray!50!white] (1.7,-0.3) rectangle (2.3,0.3) ;		
		\draw[thick, ->] (2,-0.3) -- (2,-0.7);
		\node at (2,-1) {$\sigma_{a|x}$};		
		
		\node at (0,0) {s};
		\draw[thick, dashed, color=gray!70!white] (0,0) circle (0.3cm);
		\draw[thick, dashed, color=gray!70!white, ->] (-0.3,0) -- (-1.6,0);
		\draw[thick, dashed, color=gray!70!white, ->] (0.3,0) -- (1.6,0);		
	    \end{tikzpicture}	}
	    \hspace{0.5cm}
\subfigure[\quad Sequential]{	    
		\begin{tikzpicture}[scale=0.5]
		\node at (-3.5,1.3) {Alice};
		\shade[draw, thick, ,rounded corners, inner color=white,outer color=gray!50!white] (-1.7,-0.3) rectangle (-2.3,0.3) ;
		\draw[thick, ->] (-2.5,0.5) to [out=180, in=90] (-2,0.3);
		\node at (-2.9,0.5) {$x_1$};
		\draw[thick, -<] (-2,-0.3) to [out=-90, in=180] (-2.5,-0.5);
		\node at (-3,-0.5) {$a_1$};

		\shade[draw, thick, ,rounded corners, inner color=white,outer color=gray!50!white] (-3.7,-0.3) rectangle (-4.3,0.3) ;
		\draw[thick, ->] (-4.5,0.5) to [out=180, in=90] (-4,0.3);
		\node at (-4.9,0.5) {$x_2$};
		\draw[thick, -<] (-4,-0.3) to [out=-90, in=180] (-4.5,-0.5);
		\node at (-5,-0.5) {$a_2$};

		\node at (2,1.3) {Bob};
		\shade[draw, thick, ,rounded corners, inner color=white,outer color=gray!50!white] (1.7,-0.3) rectangle (2.3,0.3) ;		
		\draw[thick, ->] (2,-0.3) -- (2,-0.7);
		\node at (2,-1) {$\sigma_{a_1a_2|x_1x_2}$};
		
		\node at (0,0) {s};
		\draw[thick, dashed, color=gray!70!white] (0,0) circle (0.3cm);
		\draw[thick, dashed, color=gray!70!white, ->] (-0.3,0) -- (-1.6,0);
		\draw[thick, dashed, color=gray!70!white, ->] (-2.35,0) -- (-3.65,0);
		\draw[thick, dashed, color=gray!70!white, ->] (0.3,0) -- (1.6,0);		
	    \end{tikzpicture}}

\subfigure[\quad Bob with input]{	    
	    \begin{tikzpicture}[scale=0.5]
		\node at (-2,1.3) {Alice};
		\shade[draw, thick, ,rounded corners, inner color=white,outer color=gray!50!white] (-1.7,-0.3) rectangle (-2.3,0.3) ;
		\draw[thick, ->] (-2.5,0.5) to [out=180, in=90] (-2,0.3);
		\node at (-2.8,0.5) {$x$};
		\draw[thick, -<] (-2,-0.3) to [out=-90, in=180] (-2.5,-0.5);
		\node at (-2.9,-0.5) {$a$};

		\node at (2,1.3) {Bob};
		\shade[draw, thick, ,rounded corners, inner color=white,outer color=gray!50!white] (1.7,-0.3) rectangle (2.3,0.3) ;		
		\draw[thick, ->] (2.5,0.5) to [out=180, in=90] (2,0.3);
		\node at (2.7,0.5) {$y$};
		\draw[thick, ->] (2,-0.3) -- (2,-0.7);
		\node at (2,-1) {$\sigma_{a|xy}$};
		
		\node at (0,0) {s};
		\draw[thick, dashed, color=gray!70!white] (0,0) circle (0.3cm);
		\draw[thick, dashed, color=gray!70!white, ->] (-0.3,0) -- (-1.6,0);
		\draw[thick, dashed, color=gray!70!white, ->] (0.3,0) -- (1.6,0);		
	    \end{tikzpicture}}	    
	    	    \hspace{0.5cm}
\subfigure[\quad Instrumental]{	    
	    \begin{tikzpicture}[scale=0.5]
		\node at (-2,1.3) {Alice};
		\shade[draw, thick, ,rounded corners, inner color=white,outer color=gray!50!white] (-1.7,-0.3) rectangle (-2.3,0.3) ;
		\draw[thick, ->] (-2.5,0.5) to [out=180, in=90] (-2,0.3);
		\node at (-2.7,0.5) {$x$};
		\draw[thick] (-2,-0.3) to [out=-30, in=-150] (1,0);
		\draw[thick, ->] (1,0) to [out=30, in=90] (2,0.3);
		\node at (-2,-0.7) {$a$};

		\node at (2,1.3) {Bob};
		\shade[draw, thick, ,rounded corners, inner color=white,outer color=gray!50!white] (1.7,-0.3) rectangle (2.3,0.3) ;		
		\node at (2,-1) {$\sigma_{a|x}$};
		\draw[thick, ->] (2,-0.3) -- (2,-0.7);
		
		\node at (0,0) {s};
		\draw[thick, dashed, color=gray!70!white] (0,0) circle (0.3cm);
		\draw[thick, dashed, color=gray!70!white, ->] (-0.3,0) -- (-1.6,0);
		\draw[thick, dashed, color=gray!70!white, ->] (0.3,0) -- (1.6,0);		
	    \end{tikzpicture}}
\end{center}
\caption{{Different generalised bipartite steering setups (a) The traditional scenario: Alice makes a measurement, steering the state of Bob. (b) The sequential-measurement scenario: Alice now performs a sequence of measurements, steering the state of Bob multiple times (c) The Bob-with-input (BWI) scenario: Bob now also has an input, allowing him to also influence his state, but performing some operation on it (d) The instrumental steering scenario: is similar to BWI except that Bob's input now depends on Alice's outcome. The top two scenarios ((a) and (b)) do not admit post-quantum steering. We show here that the bottom two scenarios ((c) and (d)) have post-quantum steering.}
\label{fig:steebip}}
\end{figure}

\textit{Preliminaries.---}
In the traditional bipartite quantum steering scenario (see Fig.~\ref{fig:steebip}(a)) Alice and Bob share a system in a possibly entangled quantum state $\rho$. 
{Alice is allowed to perform generalised measurements on her share of the system, which correspond to positive-operator valued measures (POVM). 
Alice chooses one such measurement} $\{M_{a|x}\}_a$, labelled by $x$, from a set of measurements, and obtains an outcome $a$ with probability $p(a|x) = \Tr\left\{ (M_{a|x} \otimes \id_B ) \rho \right\}$. After the measurement, Bob's steered state is $\rho_{a|x} = \Tr_A\left\{ (M_{a|x} \otimes \id_B ) \rho \right\}/p(a|x)$. It is convenient to work with the unnormalised steered states $\sigma_{a|x} = p(a|x)\, \rho_{a|x} =  \Tr_A\left\{ (M_{a|x} \otimes \id_B ) \rho \right\}$, which contain both the information about both Alice's conditional probabilities $p(a|x)=\Tr\left\{ \sigma_{a|x} \right\}$, and Bob's conditional states $\rho_{a|x}$. The collection $\{\sigma_{a|x}\}_{a,x}$ of unnormalised states Bob is steered into is called an \textit{assemblage}. Due to the completeness relation for Alice's measurements, $\sum_a M_{a|x} = \openone$ for all $x$, it follows that $\sum_a \sigma_{a|x} = \Tr_A\{\rho\} = \rho_B$, independent of $x$. This can be seen as a no-signalling condition from Alice to Bob, since Bob, without knowledge of the outcome of Alice, has no information about the choice of measurement she made. 
 
One natural generalisation {of} the traditional steering scenario is to allow Alice to make a sequence of measurements on her share of the system, such that each measurement has the potential to steer the state of Bob. This situation is depicted in Fig.~\ref{fig:steebip} (b). In the Appendix, we show that this generalisation in fact does not feature post-quantum steering either; this can be seen as an extension of the  GHJW theorem \cite{gisin,HJW} to the sequential scenario. 

\textit{Bipartite steering when Bob has an input (Fig. \ref{fig:steebip}c).---}
We consider now the generalisation where Bob's device also accepts an input before producing a quantum state. Intuitively, we can think that this input may determine the preparation of some quantum system, which could come about from a transformation on a quantum system inside Bob's device. This situation is depicted in Fig.~\ref{fig:steebip} (c), where $y$ denotes the input. In this generalised scenario, the members of the assemblage will be  $\{\sigma_{a|xy}\}_{a,x,y}$.
{Note that when the variable $y$ takes only one possible value, this scenario reduces to the traditional bipartite steering setup of Fig.~\ref{fig:steebip}(a).}

In the context of quantum theory, {we assume} that Alice and Bob share a quantum state $\rho$ and that Alice performs measurements labelled by $x$, as in the standard scenario. Given that Bob now has an input, the most general  operation that he could apply is a Completely-Positive and Trace-Preserving (CPTP) channel onto his part of the quantum system. Thus, the {quantum} assemblages that can be generated are:

\begin{defn}\textbf{Quantum Bob-{with-i}nput assemblage{s}.}\label{def:quant_gen}\\
An assemblage $\{\sigma_{a|xy}\}_{a,x,y}$ has a has a quantum realisation in the {Bob-with-input steering} scenario {iff} there exists a Hilbert space $\cH_A$ and POVMs $\{ M_{a|x} \}_{a,x}$ for Alice,  a state $\rho$ in $\cH_A \otimes \cH_B$, and a collection of CPTP maps $\{\mathcal{E}_y\}_y$ in $\cH_B$ for Bob, such that 
\begin{equation}
\sigma_{a|xy} = \mathcal{E}_y \left[ \mathrm{tr}_{A}\left\{ (M_{a|x} \otimes \id  ) \rho   \right\} \right].
\end{equation}
We denote this set of assemblages as $\mathcal{Q}_{BWI}$.
\end{defn}
{Note that this definition of Quantum Bob-with-input (BWI) assemblages does not require that the operations take place in a particular order. That is, the same assemblage can be obtained if the map $\mathcal{E}_y$ is applied to Bob's subsystem before Alice measures hers.}

To go beyond quantum theory, we have to identify the {most general} constraints that apply {here}. Not only must we now ensure no-signalling from Alice to Bob, but since Bob has an input, we must also ensure no-signalling from Bob to Alice. These constraints are captured by the following definition:

\begin{defn}\textbf{Non-signalling Bob-{with-i}nput assemblage{s}.}\label{def:axiomatic}\\ 
An assemblage $\{\sigma_{a|xy}\}_{a,x,y}$ is non-signalling in the {Bob-with-input steering scenario} i{ff} $\sigma_{a|xy} \geq 0$ for all $a,x,y$, and 
\begin{align}
\label{eq:non_oper_def}
&\sum_{a} \sigma_{a|xy} = \sum_{a} \sigma_{a|x'y}  \quad \forall \, x,x',y \,,\\
&\mathrm{tr} \left\{  \sigma_{a|xy} \right\} = p(a|x) \quad \forall \, a,x,y\,,\\
&\mathrm{tr}  \sum_{a} \sigma_{a|xy}  = 1 \quad \forall \, x,y \,,
\end{align}
where $p(a|x)$ is the probability that Alice obtains outcome $a$ when performing measurement $x$ on her share of the system. 

\noindent We denote th{e set of such} assemblages as ${\mathcal{G}}_{BWI}$.
\end{defn}

We can now return to our central question of whether there can exist post-quantum steering in this scenario. Here we find that this is indeed the {case:}
\begin{thm}
\label{theo:G_gap} The set of all non-signalling Bob-{with-i}nput assemblages is strictly larger than the set of quantum Bob-{with-i}nput assemblages, $\mathcal{Q}_{BWI} \not\equiv {\mathcal{G}}_{BWI}$. Hence, there is post-quantum steering in the Bob-{with-i}nput steering scenario.
\end{thm}
\begin{proof}
We {construct} an explicit example of a assemblage in ${\mathcal{G}}_{BWI}$ which cannot be realised in quantum theory. 

Consider the specific scenario where Alice has binary inputs and outcomes, $x \in \{0,1\}$ and $a \in \{0,1\}$, {Bob} has a binary input {$y \in \{0,1\}$,} and the dimension of Bob's Hilbert space is $2$. Consider the following assemblage: 
\begin{equation}\label{e:example BWI}
\sigma^*_{a|xy} := \frac{1}{2} \left( \ket{a}\bra{a} \, \delta_{xy=0} + \ket{a\oplus 1}\bra{a \oplus 1}  \delta_{xy=1} \right)
\end{equation}
Note that (i) $\sigma^*_{a|xy} \geq 0$ for all $a,x,y$; (ii) $\sum_a \sigma^*_{a|x,y} = \frac{1}{2}(\delta_{xy=0} + \delta_{xy=1}) \id = \frac{1}{2}\id$, which is independent of $x$ and $y$; (iii) $\Tr \left\{ \sigma^*_{a|xy} \right\} = \frac{1}{2}(\delta_{xy=0} + \delta_{xy=1}) = \frac{1}{2}$ is independent of $y$ and (iv) $\Tr \sum_a \sigma_{a|xy}^* = 1$. This shows that $\{\sigma_{a|xy}^*\}$ is a valid no-signalling assemblage, i.e., $\{\sigma^*_{a|xy}\} \in {\mathcal{G}}_{BWI}$. 

Now we show that this assemblage cannot arise in quantum theory, {i.e.,} $\{\sigma^*_{a|xy}\} \not\in \mathcal{Q}_{BWI}$. We do so by first noting that for a quantum-realisable assemblage, since $(\id_A \otimes \mathcal{E}^y) [\rho]$ is a bipartite quantum state when $\mathcal{E}^y$ is a CPTP channel, Alice and Bob {can only} produce quantum {Bell} correlation{s,} should Bob choose to measure his system. Namely, let Bob make an arbitrary measurement $\{N_b\}_b$, on his state in a quantum assemblage $\{\sigma_{a|xy}\}_{a,x,y}$. Then, the correlations obtained are 
\begin{align}
	p(a,b|x,y) &= \Tr\{N_b \sigma_{a|xy}\}, \nonumber \\ \nonumber
	&= \Tr_B \{N_b\mathcal{E}_y \left[ \mathrm{tr}_{A}\left\{ (M_{a|x} \otimes \id  ) \rho   \right\} \right]\}, \\
	&= \Tr\left\{ (M_{a|x} \otimes \mathcal{E}^\dagger_y(N_b)  ) \rho   \right\}. \nonumber
\end{align}
Since $\mathcal{E}_y(\cdot)$ is a CPTP channel , the dual map $\mathcal{E}_y^\dagger(\cdot)$ is unital, and hence $\mathcal{E}_y^\dagger(N_b)$ is always valid POVM. This provides an explicit quantum realisation of the correlations $p(a,b|x,y)$. We will thus prove that \eqref{e:example BWI} is not quantum-realisable by demonstrating that it can generate correlations $p(a,b|x,y)$ which are known to be impossible within quantum theory. 

Let $N_b = \ket{b}\bra{b}$ be the computational basis measurement. The correlations that Alice and Bob obtain are 
\begin{align*}
p(a,b|x,y) =  \bra{b} \sigma^*_{a|xy} \ket{b} = \begin{cases}
\frac{1}{2} \quad \text{if} \quad a\oplus b = xy \\
0 \quad \text{otherwise} 
\end{cases}\,.
\end{align*}
These are the correlations of the "Popescu-Rohrlich'' box \cite{PR}, which are not achievable within quantum theory. Hence, $\sigma^*_{a|xy} \not\in \mathcal{Q}_{BWI} $ and so $\mathcal{Q}_{BWI} \not\equiv {\mathcal{G}}_{BWI}$.  
\end{proof}

We see then that post-quantum steering can arise in a generalised bipartite steering scenario. This example given however relies on post-quantum nonlocality and hence the post-quantum steering found may be argued to be just another guise of the former effect. In the following theorem, we prove that the two phenomena are genuinely different:

\begin{thm}
\label{thm:ex}
Post-quantum steering in the Bob-{with-i}nput steering scenario is independent of post-quantum nonlocality. Namely, there exist non-signalling assemblages $\{\sigma_{a|xy}\}$ that are not quantum realisable, but which can only lead to quantum correlations $p(a,b|x,y)$ in the Bell scenario.
\end{thm}
The proof of this theorem is given in the Appendix. The main idea is to show that the following assemblage has the desired properties, i.e., that it is post-quantum and that whenever Bob performs a measurement $\{N_b\}$ on it, the observed outcome statistics $p(ab|xy) = \Tr\left\{ N_b \, \sigma_{a|xy} \right\} $ have always a quantum realisation:
\begin{align}\label{e:example}
\sigma_{a|xy} = \frac{1}{4}(\openone + (-1)^{a+\delta_{x,2}\delta_{y,1}} \sigma_x )
\end{align}
where $x \in \{1,2,3\}$ and  $(\sigma_1,\sigma_2,\sigma_3) = (X,Y,Z)$ are the Pauli operators.

The method to show that this assemblage can only yield quantum correlations is to notice that one may mathematically represent this assemblage as
Alice performing Pauli measurements on the maximally entangled state, and Bob applying either the identity or transpose map (which crucially is positive but not completely positive) depending on $y$ \cite{footnote1}. Then, following Ref.~\cite{oform}, the assemblage Eq.~\eqref{e:example} can only yield quantum correlations. 
{In the Appendix we further prove that the assemblage is postquantum: we construct an explicit steering inequality that is (robustly) violated by the assemblage \eqref{e:example} beyond the quantum bound. For this, we present a method to bound the quantum bound of a steering inequality in this scenario. We also present an alternative analytic proof of the claim.}

Hence, just as for multipartite post-quantum steering \cite{pqs}, the effect here is independent of the existence of post-quantum Bell nonlocality.

\textit{Instrumental steering  (Fig. \ref{fig:steebip}d).---}
We now consider the instrumental steering scenario \cite{inststee}. In this case, Bob still {has an input that can inform the preparation of a quantum system}, however now {this input can} depend on Alice's measurement outcome (see Fig.~\ref{fig:steebip} (d)). {For example, Bob's input could just decide a transformation upon a quantum system.} {To recover the traditional steering scenario, we again enforce the constraint that Bob only has one input, and thus we trivially have no dependence on Alice's output.} This scenario is closely related to the so-called `instrumental setup' \cite{SGS01,Pearl09}, only now one of the variables has become a quantum system. Indeed, this close relation between instrumental steering and the instrumental setup, will enable us to identify a connection between the instrumental steering scenario and the Bob-with-input scenario further below.

In the instrumental steering {scenario,} an assemblage is given by the collection of subnormalised states $\{\sigma_{a|x}\}$, where $x$ denotes the choice of measurement by Alice, and $a$ denotes both Alice's outcome and Bob's input. Within quantum theory, the assemblages they can generate {are the} following:

\begin{defn}\textbf{Quantum Instrumental assemblage{s}.}\\
An assemblage $\{\sigma_{a|x}\}_{a,x}$ has a has a quantum realisation in the instrumental steering scenario {iff} there exists a Hilbert space $\cH_A$ and POVMs $\{ M_{a|x} \}_{a,x}$ for Alice, a state $\rho$ in $\cH_A \otimes \cH_B$, and a collection of CPTP maps $\{\mathcal{E}_a\}_a$ in $\cH_B$ for Bob, such that
\begin{equation}
\sigma_{a|x} = \mathcal{E}_a \left[ \Tr_{A}\left\{ (M_{a|x} \otimes \id) \rho   \right\} \right].
\end{equation} 
We denote this set of assemblages by $\mathcal{Q}_I$.
\end{defn}

The instrumental steering scenario has no straightforward non-signalling constraints.
Hence, in order to define general assemblages here, we adopt the relation between non-signalling Bell correlations and generic instrumental correlations in the black-box scenario found in Ref.~\cite{pironio} {(see also Supplementary Material of Ref.~\cite{chaves})}. {In the instrumental setup, where the so-called device-independent instrumental correlations are studied, it was recently found that these correlations are indeed a post-selection of the correlations found in a Bell scenario: the post-selection procedure consists on keeping the events where $y=a$ \cite{pironio}. This inspires the following definition:}

\begin{defn}\label{def:steeint}\textbf{{General} instrumental assemblage{s}.}\\
An assemblage $\{\sigma_{a|x}\}_{a,x}$ is a {general} instrumental assemblage i{ff} there exists a non-signalling Bob-{with-i}nput assemblage 
$\{\omega_{a|xy}\} \in {\mathcal{G}}_{BWI}$ such that
$ \sigma_{a|x} = \omega_{a|x,y=a}$ {for all $a$ and $x$}.
\noindent We denote th{e set of such general} assemblages by ${\mathcal{G}}_I$.
\end{defn}
{This definition allows us to adopt the viewpoint of Refs.~\cite{pironio,chaves}, and hence understand the assemblages in the instrumental steering scenario as being a post-selection of those in a Bob-with-input scenario.}

{Note that this connection between the Bob-with-input scenario and the instrumental steering scenario allows us to interpret the latter beyond the traditional way that the instrumental setup is presented. Usually, the instrumental setup is such that there is signalling from Alice to Bob, since he needs to learn her outcome in order to implement the operation on his system. However, the particular perspective brought in by Ref.~\cite{pironio}, and which we adopt here, highlights that, ultimately, this communication plays no distinct role in how resourceful the assemblages are, since the Bob-with-Input scenario does not allow for signalling and can simulate them.}

{Returning to} our central question, we now show that there is post-quantum steering in the instrumental steering scenario. Moreover, we show that this does not follow from post-quantum instrumental {black-box} correlations, and {it} i{s} hence another independent form of post-quantumness. 
\begin{thm}
\label{theo:inst1} The set of {general} instrumental assemblages {strictly contains} the set of quantum instrumental assemblages, $\mathcal{Q}_{I} \not\equiv {\mathcal{G}}_{I}$. Hence, {post-quantum instrumental steering exists}.
\end{thm}

\begin{thm}
\label{thm:inst2}
Post-quantum steering in the instrumental steering scenario is independent of post-quantum instrumental correlations. Namely, there exist {general} assemblages $\{\sigma_{a|x}\}$ that are not quantum realisable, but which can only lead to quantum correlations $p(a{,b}|x)$ in the instrumental scenario.
\end{thm}
T{hese two theorems are proven} together in the Supplementary Material, but {their proof is} very similar to th{at} of Theorem \ref{thm:ex}. The {general} assemblage that is used {here as an example} is that which derives from \eqref{e:example} by setting $y = a$, which is both provably post-quantum in the instrumental scenario, and can only lead to quantum instrumental {black-box} correlations.

{Thus,} post-quantum steering is also possible within the instrumenta{l s}cenario, an{d t}his is {independent of the existence of} correlations with no quantum explanation in the fully device-independent instrumental scenario. Hence, post-quantum instrumental steering is another genuinely new {effect. 
Finally,} in terms of number of variables (inputs and outputs), the instrumental scenario is the simplest one where post-quantum steering can exist.

\textit{Discussion.---}
Exploring plausible effects beyond quantum theory that are nevertheless consistent with relativistic causality  \cite{PR}, is important from various perspectives: on the one hand, it allows possible extensions of quantum theory to be explored, in the light of quantum gravity \cite{hardy}. On the other, it allows us to develop a deeper understanding of quantum theory itself, by identifying those properties of it that are truly quantum \cite{hardy2,giulio,lluis,john}. Here, we have shown that for the important form of nonlocality known as steering, it is possible in principle to go beyond what quantum theory allows even when considering only two parties, if suitable generalisations of the traditional scenario are considered. Crucially, we showed that our examples of post-quantum steering are genuinely new, and are not related to other post-quantum nonlocal effects. {We further developed a numerical technique to bound the maximum quantum violation of a steering inequality --  a promising tool for developing future applications of this new type of steering.}
 
In addition, on the way, we also showed that post-quantum steering is impossible in the sequential-measurement generalisation of steering, schematically depicted in Fig.~\ref{fig:steebip} (b).  As such, we have extended the GHJW no-go theorem \cite{gisin, HJW} to this setting, with detailed provided in the Appendix.

The `instrumental setup' \cite{SGS01,Pearl09} is known to be the one with the fewest number of variables able to admit a classical-quantum gap \cite{HLP14}. This is closely related to the setup of Fig.~\ref{fig:steebip}(d), {except that} Bob's system is a classical variable. Previously, classical-quantum gaps had been found in Bell-type \cite{chaves,pironio,Agresti} and steering \cite{inststee} scenarios. Furthermore, quantum-post-quantum gaps have also been found in Bell-type scenarios \cite{chaves,pironio}; but the existence of post-quantum instrumental steering remained an open question. {The discovery of the latter here thus also resolves this open question.}

{Going forward, the most interesting question now is to understand the power of post-quantum steering. For instance, are there information-theoretic or physical principles that are violated by the newly-discovered forms of post-quantum steering found here? In addition, it would be interesting to explore information processing tasks exploiting post-quantum steering as a resource. For example, one task where traditional steering is a resource is subchannel discrimination \cite{pianiwatrous}. It would be interesting to study whether post-quantum steering gives an advantage in tasks related to this. More generally, now that we have uncovered post-quantum steering in a bipartite setting, it paves the way for analysing a broad range of bipartite tasks from this new direction. Indeed, we note that our newly introduced Bob-with-input steering scenario has already been investigated within the context of resource theories \cite{Schmid19}. {We expect our newly developed numerical technique will become relevant for exploring the power of post-quantum steering.}

Our overarching hope is that studying quantum theory `from the outside', whether from the perspective of steering, or other nonlocal and nonclassical effects, will lead to novel insights into the very structure of quantum theory and the possibilities and limitations of quantum theory for information processing. We expect our results and new insights to contribute to this rapidly developing and exciting field.}

\textit{Acknowledgements.--}
We thank {Rodrigo Gallego,} Matt Pusey, John Selby and Elie Wolfe for fruitful discussion{s}.
This research was supported by Perimeter Institute for Theoretical Physics. Research at Perimeter Institute is supported by the Government of Canada through the Department of Innovation, Science and Economic Development Canada and by the Province of Ontario through the Ministry of Research, Innovation and Science.
ABS acknowledges partial support by the Foundation for Polish Science (IRAP project, ICTQT, contract no. 2018/MAB/5, co-financed by EU within Smart Growth Operational Programme). 
MJH and ABS acknowledge the FQXi large grant ``The Emergence of Agents from Causal Order''. 
LA acknowledges financial support from the Brazilian agencies CNPq (PQ grant No. 311416/2015-2 and INCT-IQ), FAPERJ (JCN E-26/202.701/2018), CAPES (PROCAD2013), and the Serrapilheira Institute (grant number Serra-1709-17173). PS acknowledges support from a Royal Society URF (UHQT).

\onecolumngrid
\appendix

\section{Review of the Gisin's and Hughston-Josza-Wooter's theorem}

In this section we review the proof of the theorem by Gisin \cite{gisin} and Hughston, Josza and Wootters \cite{HJW} (GHJW), which states that post-quantum steering in bipartite traditional scenarios compatible with the No Signalling (NS) principle does not exist. 

The idea is to find a quantum realization of a generic NS assemblage. One could view this NS assemblage as a non-signalling sequential assemblage (as defined in the main body of the paper) with only one round of measurements. Given an NS assemblage $\{\sigma_{a|x}\}$, we want to find a Hilbert space for Alice, a state $\rho$ shared by Alice and Bob, and measurement operators $\{M_{a|x}\}$ for Alice such that: 
\begin{align*}
\begin{cases}
\sum_a M_{a|x} = \id_A \quad \forall \, x\,,\\
\sigma_{a|x} = \mathrm{tr}_A \left\{ (M_{a|x} \otimes \id_B) \, \rho \right\} \quad \forall \, a,x\,.
\end{cases}
\end{align*}

GHJW find such a general construction as follows. First, notice that {the reduced state $\sigma_R = \sum_a \sigma_{a|x}$ satisfies} $\sigma_R \geq 0$, hence it has a {diagonal} decomposition as: 
\begin{align*}
\sigma_R = \sum_{r \in S_R} \mu_r \ket{r}\bra{r}\,,
\end{align*}
where $\{\ket{r}\}_r$ is an orthonormal basis, and $S_R$ is the set of those basis' vectors that $\sigma_R$ has support on (that is, we can assume $\mu_r > 0$).
GHJW then define the following state and measurements for the quantum realisation of the given assemblage: 
\begin{align}
\label{eq:state}
\ket{\psi} &= \sum_{r \in S_R} \sqrt{\mu_r} \ket{r} \otimes \ket{r} \,,\\
\label{eq:def_POVM}
M_{a|x} &= \begin{cases} \sqrt{\sigma_R^{-1}} \, \sigma_{a|x}^{\mathrm{T}} \, \sqrt{\sigma_R^{-1}}\,, \quad \text{if} \quad a>0 \\
	\sqrt{\sigma_R^{-1}} \, \sigma_{a|x}^{\mathrm{T}} \, \sqrt{\sigma_R^{-1}} + \id -  \id_R \,, \quad \text{if} \quad a=0
	\end{cases} \,,
\end{align}
where 
\begin{align*}
\sqrt{\sigma_R^{-1}} = \sum_{r \in S_R} \frac{1}{\sqrt{\mu_r}} \ket{r}\bra{r}\,,
\end{align*}
and $ \id_R = \sum_{r \in S_R} \ket{r}\bra{r}$.
These state and measurements are a valid quantum realisation of the assemblage, as we see next. 

Let's first show that the collection of operators $\{M_{a|x}\}$  forms well defined measurements. Since the operators are positive semidefinite by definition, we only need to check that the are suitably normalised: 
\begin{align*}
\sum_a M_{a|x} &= \sqrt{\sigma_R^{-1}} \, \sigma_R^{\mathrm{T}} \, \sqrt{\sigma_R^{-1}} + \id -  \id_R \\
&= \sum_{r,r',r" \in S_R} \frac{\mu_{r'}}{\sqrt{\mu_r}\sqrt{\mu_{r"}}} \ket{r}\bra{r}\ket{r'}\bra{r'}\ket{r"}\bra{r"} + \id -  \id_R \\
&= \sum_{r \in S_R} \ket{r}\bra{r} + \id -  \id_R \\
&= \id \,. \nonumber
\end{align*}

Now let's show that the model actually recovers the assemblage. When $a>0$:
\begin{align} \nonumber
\mathrm{tr}_A \left\{ (M_{a|x} \otimes \id_B) \, \rho \right\} &= \mathrm{tr}_A \left\{ \left(\sqrt{\sigma_R^{-1}} \, \sigma_{a|x}^{\mathrm{T}} \, \sqrt{\sigma_R^{-1}} \otimes \id_B \right) \, \left(\sum_{r,r' \in S_R} \sqrt{\mu_r}\sqrt{\mu_{r'}} \ket{rr} \bra{r'r'} \right) \right\} \\ \nonumber
&= \mathrm{tr}_A \left\{ \left( \sigma_{a|x}^{\mathrm{T}} \otimes \id_B \right) \, \left(\sum_{r,r' \in S_R}  \ket{rr} \bra{r'r'} \right) \right\} \\
&= \sigma_{a|x} \,,\nonumber
\end{align}
where to see that these equalities follow we made use of some technical lemmas that you may find in the {next} section of our supplemental material: by Lemma \ref{lem:sup}, {the support of $\sigma_{a|x}$ is not outside that of $\sigma_R$}, and hence Lemma \ref{lem:tel} applies. 

Now, when $a=0$:
\begin{align} \nonumber
\mathrm{tr}_A \left\{ (M_{a|x} \otimes \id_B) \, \rho \right\} &= \mathrm{tr}_A \left\{ \left(\sqrt{\sigma_R^{-1}} \, \sigma_{a|x}^{\mathrm{T}} \, \sqrt{\sigma_R^{-1}} \otimes \id_B \right) \, \left(\sum_{r,r' \in S_R} \sqrt{\mu_r}\sqrt{\mu_{r'}} \ket{rr} \bra{r'r'} \right) \right\}  \\ \nonumber
&+ \mathrm{tr}_A \left\{ \left( (\id -  \id_R) \otimes \id_B \right) \, \left(\sum_{r,r' \in S_R} \sqrt{\mu_r}\sqrt{\mu_{r'}} \ket{rr} \bra{r'r'} \right) \right\} \\  \nonumber
&= \mathrm{tr}_A \left\{ \left( \sigma_{a|x}^{\mathrm{T}} \otimes \id_B \right) \, \left(\sum_{r,r' \in S_R}  \ket{rr} \bra{r'r'} \right) \right\} \\ \nonumber
&= \sigma_{a|x} \,.
\end{align}

This completes the proof of the GHJW theorem for bipartite steering scenarios.


\section{The sequential-measurement scenario admits no post-quantum steering}

One natural generalisation {of} the traditional steering scenario is to allow Alice to make a sequence of measurements on her share of the system, such that each measurement has the potential to steer the state of Bob. For clarity in the presentation, we will focus on the case depicted in Fig.~\ref{fig:steebip} (b), where Alice makes two measurements, with general case of an arbitrary number of measurements following by induction{. 
I}n this scenario, Alice first chooses to make a measurement labelled by $x_1$, with outcomes labelled by $a_1$. Since we will need the post-measurement state, we must specify the Kraus operators $\{K_{a_1|x_1}\}_{a_1}$ such that $M_{a_1|x_1} = K^\dagger_{a_1|x_1}\,{K_{a_1|x_1}}$, and not just the POVM elements $M_{a_1|x_1}$. Alice then chooses to perform a second measurement, labelled by $x_2$, with outcome $a_2$, with corresponding POVM elements $\{M_{a_2|x_2}\}_{a_2}$ \cite{FTN}. In this case, {Bob's assemblage has} elements $\sigma_{a_1a_2|x_1x_2} = \mathrm{tr}_{A}\left\{ (K^\dagger_{a_1|x_1} \, M_{a_2|x_2} K_{a_1|x_1} \otimes \id_B) \rho  \right\}$ such that the probabilities for Alice's pair of measurement {outcomes} are $p(a_1,a_2|x_1,x_2) = \Tr \left\{ \sigma_{a_1a_2|x_1x_2} \right\}$. This leads us to the following definition: 

\begin{defn}\textbf{Quantum sequential assemblage{s.}}\\
An assemblage $\{\sigma_{a_1a_2|x_1x_2}\}_{a_1,a_2,x_1,x_2}$ has a quantum realisation in the sequential steering {scenario if} there exists a Hilbert space $\cH_A$ for Alice, Kraus operators $\{ K_{a_1|x_1} \}_{a_1,x_1}$ and POVMs $\{M_{a_2|x_2}\}_{a_2,x_2}$ for Alice, and a state $\rho$ in $\cH_A \otimes \cH_B$, such that 
\begin{align}\label{eq:theqseq}
\sigma_{a_1a_2|x_1x_2} &= \mathrm{tr}_{A}\left\{ (K^\dagger_{a_1|x_1} \, M_{a_2|x_2} K_{a_1|x_1} \otimes \id_B ) \rho  \right\}.
\end{align}
We denote this set of assemblages as $\mathcal{Q}_S$.
\end{defn}
We are then interested in thinking about this scenario more abstractly than from the point of view of quantum {theory: we} would like to think about the most general type of correlation that could arise in this sequential scenario, potentially beyond what quantum theory predicts. For {this, we} need to find the {corresponding} no-signalling conditions that apply. 
{Here, we} still have no-signalling from Alice to Bob, just as before, but now taking into account all the information on Alice's {side: that is,} $\sum_{a1,a_2} \sigma_{a_1a_2|x_1x_2}$ is independent of $x_1$ and $x_2$. {Moreover, we} have an additional no-signalling constraint, from the future to the past within Alice's lab: {the outcome statistics of the first measurement performed by Alice ($x_1$) should not depend on her choice of measurement for the future time step ($x_2$)}. This means that $\sum_{a_2} \sigma_{a_1 a_2|x_1 x_2}$ must be independent of $x_2$. We thus {naturally arrive at} the following definition for the most general non-signalling assemblage{.}
\begin{defn}\textbf{Non-signalling sequential assemblage{s.}}\\
An assemblage $\{\sigma_{a_1a_2|x_1x_2}\}_{a_1,a_2,x_1,x_2}$ is non-signalling in the {sequential steering scenario iff} $\sigma_{a_1 a_2 |x_1 x_2} \geq 0$ for all $a_1$,$a_2$,$x_1$,$x_2$,  and 
\begin{align}
&\sum_{a_1,a_2} \sigma_{a_1a_2|x_1x_2} &= \sum_{a_1,a_2} \sigma_{a_1a_2|x'_1x'_2} \quad \forall \, x_1, x'_1, x_2, x'_2,\\
&\sum_{a_2} \sigma_{a_1a_2|x_1x_2} &= \sum_{a_2} \sigma_{a_1a_2|x_1x'_2} \quad \forall \, a_1, x_1, x_2, x_2'.
\end{align}
We denote th{e} set of {such} assemblages as ${\mathcal{G}}_S$.
\end{defn}
{The choice of notation ${\mathcal{G}}_S$ comes from thinking of these assemblages as the most general ones compatible with the no-signalling constraints of the setup.}

The first question w{e c}onsider is the{n} whether every {non-signalling} sequential assemblage is also a quantum sequential {assemblage. 
As} a first main result, we prove that this is not the cas{e:}
\begin{thm}
\label{theo:seq}
The set of all non-signalling sequential assemblages coincides with the set of quantum sequential assemblages, ${\mathcal{Q}}_S = \mathcal{G}_S$. {That is, there} is no post-quantum steering in the bipartite sequential steering scenario.
\end{thm}
The basic idea in the proof of this theorem is to extend the standard GHJW construction working sequentially through the measurements of Alice in the order they are performed, using  the L\"uders rule to specify their Kraus operators. 
Before proceeding to the formal proof, we need three technical lemmata. The first lemma concerns a traditional bipartite steering scenario, for an assemblage of the form $\{\sigma_{a|x}\}_{a,x}$:
\begin{lemma}\label{lem:1}
{The support of $\sigma_{a|x}$ is equal to or contained in that of} $\sigma_R$, for all $a,x$.
\end{lemma}
\begin{proof}
Since $\sigma_R = \sum_\alpha \sigma_{\alpha|x}$, then $\sigma_R - \sigma_{a|x}= \sum_{\alpha \neq a} \sigma_{\alpha|x} \geq 0$.

Now let $\id_{\overline{R}}$ be the projector onto the {orthogonal complement to the support of} $\sigma_R$. Then, by definition $\id_{\overline{R}} \sigma_R \id_{\overline{R}} = 0$. Now let's assume that {the support of $\sigma_{a|x}$ is not contained in that of $\sigma_R$}. That is, $\id_{\overline{R}} \sigma_{a|x} \id_{\overline{R}} > 0$. Then,
\begin{align*}
\id_{\overline{R}} \left( \sigma_R - \sigma_{a|x} \right)  \id_{\overline{R}} = - \id_{\overline{R}} \sigma_{a|x} \id_{\overline{R}} < 0\,.
\end{align*} 
which is a contradiction. Hence, $\sigma_{a|x}$ must live in the support of $\sigma_R$. 
\end{proof}

The second lemma concerns a sequential bipartite steering scenario, for an assemblage $\{\sigma_{a_{1}a_{2}|x_{1}x_{2}}\}_{a_{1},a_{2},x_{1},x_{2}} \in {\mathcal{G}}_s$.
\begin{lemma} \label{lem:sup} 
\begin{enumerate}
\item[(i)] The marginal state $\sigma_{a_{1}|x_{1}}$ lives in the subspace supported by $\sigma_R$, for all $a,x$. 
\item[(ii)] {The support of $\sigma_{a_{1}a_{2}|x_{1}x_{2}}$ is equal to or contained in that of}  $\sigma_{a_{1}|x_{1}}$, for all $a_{1},a_{2},x_{1},x_{2}$. 
\end{enumerate}
\end{lemma}
\begin{proof}
The proof of (i) follows similarly to that of Lemma \ref{lem:1} once we see that when considering only the marginal states the situation is equivalent to the traditional bipartite steering scenario.

The proof of (ii) follows a similar logic than that of Lemma \ref{lem:1}, by having $\sigma_{a_{1}a_{2}|x_{1}x_{2}}$ play the role of $\sigma_{a_{1}|x_{1}}$  and $\sigma_{a_{1}|x_{1}}$ that of $\sigma_{R}$. In the following we present the proof explicitly. For simplicity we will replace $a_{1}$ and $a_{2}$ with $a$ and $b$ respectively, and $x_{1}$ and $x_{2}$ with $x$ and $y$ respectively.

Since $\sigma_{a|x} = \sum_\beta \sigma_{a\beta |xy}$, then $\sigma_{a|x} - \sigma_{ab|xy}= \sum_{\beta \neq b} \sigma_{a\beta | xy} \geq 0$.
Now let $\id_{\overline{ax}}$ be the projector onto the space not supported by $\sigma_{a|x}$. Then, by definition $\id_{\overline{ax}} \sigma_{a|x} \id_{\overline{ax}} = 0$. Now let's assume that $\sigma_{ab|xy}$ has support outside the space where $\sigma_{a|x}$ lives. That is, $\id_{\overline{ax}} \sigma_{ab|xy} \id_{\overline{ax}} > 0$. Then,
\begin{align*}
\id_{\overline{ax}} \left( \sigma_{a|x} - \sigma_{ab|xy} \right)  \id_{\overline{ax}} = - \id_{\overline{ax}} \sigma_{ab|xy} \id_{\overline{ax}} < 0\, ,
\end{align*} 
which is a contradiction. Hence, $\sigma_{ab|xy}$ must live in the support of $\sigma_{a|x}$. 
\end{proof}

The third, and final, lemma is a special instance of the link product \cite{chlink} in the Choi representation, and we make it explicit for simplicity in the before mentioned proofs. 
\begin{lemma}\label{lem:tel}
\begin{align*}
\mathrm{tr}_1 \left\{ (A^{\mathrm{T}} \otimes \id_2) \ket{\psi}\bra{\psi}  \right\} = A \,,
\end{align*}
where $\ket{\psi} = \sum_k \ket{k} \otimes \ket{k}$, and $\{\ket{k}\}_k$ spans the space where $A$ has support on.
\end{lemma}
\begin{proof}
Let $d$ be the dimension of the Hilbert space where $A$ lives, and $n$ the dimension of the space spanned by $\{\ket{k}\}_k$. 
\begin{align*}
\mathrm{tr}_1 \left\{ (A^{\mathrm{T}} \otimes \id_2) \ket{\psi}\bra{\psi}  \right\} &= \sum_{j=1}^d \sum_{k,l=1}^n \bra{j} A^{\mathrm{T}} \ket{k} \langle l \vert j \rangle \ket{k}\bra{l}
= \sum_{j=1}^n \sum_{k=1}^n A^{\mathrm{T}}_{jk} \ket{k}\bra{j} 
= \sum_{j,k=1}^d  A_{kj} \ket{k}\bra{j} = A\,.
\end{align*}
\end{proof}

We are now in a position to prove Theorem \ref{theo:seq}.
\\
\\
\textit{Proof of Theorem 1.} Let us start from an assemblage $\{ \sigma_{ab|xy} \} \in {\mathcal{G}}_s$, where for simplicity we change slightly the notation and denote by $(b,y)$ the labels of the second time step, and $(a,x)$ those of the first time step. Now define the following: 
\begin{align} \label{eq:1}
\ket{\psi} &= \sum_{r \in S_R} \sqrt{\mu_r} \ket{r} \otimes \ket{r} \,,\\ \label{eq:2}
K_{a|x} &= \begin{cases} \sqrt{\sigma_{a|x}^{\mathrm{T}}} \, \sqrt{\sigma_R^{-1}} \quad \text{if} \quad a>0\\
\sqrt{\sigma_{a|x}^{\mathrm{T}}} \, \sqrt{\sigma_R^{-1}} + \id - \id_{\overline{R}} \quad \text{if} \quad a=0
\end{cases} \,, \\ \label{eq:3}
M^{a,x}_{b|y} &= \begin{cases} \left(\sqrt{\sigma_{a|x}^{\mathrm{T}}}\right)^{-1} \, \sigma_{ab|xy}^{\mathrm{T}} \, \left(\sqrt{\sigma_{a|x}^{\mathrm{T}}}\right)^{-1} \quad \text{if} \quad b>0 \\
 \left(\sqrt{\sigma_{a|x}^{\mathrm{T}}}\right)^{-1} \, \sigma_{ab|xy}^{\mathrm{T}} \, \left(\sqrt{\sigma_{a|x}^{\mathrm{T}}}\right)^{-1} + \id - \id_{\overline{ax}} \quad \text{if}  \quad
b=0
\end{cases} \,,
\end{align}
where: 
\begin{itemize}
\item $\sigma_R = \sum_{r \in S_R} \mu_r \ket{r}\bra{r}$, 
with $\{\ket{r}\}_r$ an orthonormal basis, and $S_R$ the set of those basis' vectors that $\sigma_R$ has support on.
\item $\sqrt{\sigma_R^{-1}} = \sum_{r \in S_R} \frac{1}{\sqrt{\mu_r}} \ket{r}\bra{r}$.
\item $\sigma_{a|x} = \sum_{k \in S_{a,x}} \nu^{a,x}_k \ket{k}_{a,x}\bra{k}_{a,x}$ ,
with $\{\ket{k}_{a,x}\}_k$ an orthonormal basis, and $S_{a,x}$ the set of those basis' vectors where $\sigma_{a|x}$ has support on.
\item $\sqrt{\sigma_{a|x}^{\mathrm{T}}} = \sum_{k \in S_{a,x}} \sqrt{\nu^{a,x}_k} \ket{k}_{a,x}\bra{k}_{a,x}$\, , $\left(\sqrt{\sigma_{a|x}^{\mathrm{T}}}\right)^{-1} = \sum_{k \in S_{a,x}} \frac{1}{\sqrt{\nu^{a,x}_k}} \ket{k}_{a,x}\bra{k}_{a,x}$ .
\item $ \id_{\overline{R}} $ the projector onto the orthogonal complement of the support of $\sigma_R$,
\item $ \id_{\overline{ax}} $ the projector onto the orthogonal complement of the support of $\sigma_{a|x}$,
\end{itemize}

Now we see how the state, measurement operators and Kraus operators reproduce the assemblage and give well defined measurements for each time step. 

Let us begin showing that the Kraus operators $\{K_{a|x}\}$ correspond to well defined measurements for the first time step. First, $\sum_a K_{a|x}^\dagger K_{a|x} \geq 0$ by definition. Second,
\begin{align} \nonumber
\sum_a K_{a|x}^\dagger K_{a|x} &= \sum_{a \neq 0}  \sqrt{\sigma_R^{-1}} \, \sigma_{a|x}^{\mathrm{T}} \, \sqrt{\sigma_R^{-1}} + \left(  \sqrt{\sigma_R^{-1}} \, \sqrt{ \sigma_{0|x}^{\mathrm{T}}} + \id_{\overline{R}} \right) \left( \sqrt{ \sigma_{0|x}^{\mathrm{T}}} \, \sqrt{\sigma_R^{-1}} + \id_{\overline{R}} \right) \\ \nonumber
&= \sum_{a \neq 0}  \sqrt{\sigma_R^{-1}} \, \sigma_{a|x}^{\mathrm{T}} \, \sqrt{\sigma_R^{-1}} + \sqrt{\sigma_R^{-1}} \, \sigma_{0|x}^{\mathrm{T}} \, \sqrt{\sigma_R^{-1}} +  \id_{\overline{R}} \\ \nonumber
&= \sqrt{\sigma_R^{-1}} \, \sigma_R \, \sqrt{\sigma_R^{-1}} +  \id_{\overline{R}} = \id_R +  \id_{\overline{R}}\\ \nonumber
&= \id\,,
\end{align}
where we used the fact that $ \id_{\overline{R}} \, \sigma_{a|x}^{\mathrm{T}} = 0 = \sigma_{a|x}^{\mathrm{T}} \,  \id_{\overline{R}} $ according to Lemma \ref{lem:sup}. These two properties show that the Kraus operators correspond indeed to a well defined measurement. 

Now let us show that the operators $\{M_{b|y}^{a,x}\}$ define a measurement for each $a,x,y$ in the second time step. First, notice that each $M_{b|y}^{a,x}$ is positive semidefinite by definition, hence we only need to check that they are properly normalised. This is shown as follows: 
\begin{align} \nonumber
\sum_b M_{b|y}^{a,x} &= \sum_b \left(\sqrt{\sigma_{a|x}^{\mathrm{T}}}\right)^{-1} \, \sigma_{ab|xy}^{\mathrm{T}} \, \left(\sqrt{\sigma_{a|x}^{\mathrm{T}}}\right)^{-1}  +  + \id - \id_{\overline{ax}} \\ \nonumber
&= \left(\sqrt{\sigma_{a|x}^{\mathrm{T}}}\right)^{-1} \, \sigma_{a|x}^{\mathrm{T}} \, \left(\sqrt{\sigma_{a|x}^{\mathrm{T}}}\right)^{-1} +  + \id - \id_{\overline{ax}} \\ \nonumber
&= \sum_{k,k',k" \in S_{a,x}} \frac{\nu_{k'}^{a,x}}{\sqrt{\nu_{k}^{a,x}}\sqrt{\nu_{k"}^{a,x}}} \ket{k}_{a,x}\bra{k}_{a,x} \ket{k'}_{a,x}\bra{k'}_{a,x} \ket{k"}_{a,x}\bra{k"}_{a,x}  + \id - \id_{\overline{ax}}  \\ \nonumber
&= \sum_{k \in S_{a,x}} \ket{k}_{a,x}\bra{k}_{a,x}  + \id - \id_{\overline{ax}}  \\ \nonumber
&= \id\, .
\end{align}

Finally,  let us show that these state and measurements recover the assemblage. Here we will use some technical lemmas presented earlier. Let us begin with the case of $a>0$ and $b>0$: 
\begin{align*} \nonumber
&\mathrm{tr}_{A}\left\{ K^\dagger_{a|x} \,  M^{ax}_{b|y} K_{a|x} \otimes \id_B \, \rho  \right\}  \\ &= \mathrm{tr}_{A}\left\{ \left( \sqrt{\sigma_R^{-1}} \, \sqrt{\sigma_{a|x}^{\mathrm{T}}} \, \left(\sqrt{\sigma_{a|x}^{\mathrm{T}}}\right)^{-1} \, \sigma_{ab|xy}^{\mathrm{T}} \, \left(\sqrt{\sigma_{a|x}^{\mathrm{T}}}\right)^{-1} \sqrt{\sigma_{a|x}^{\mathrm{T}}} \, \sqrt{\sigma_R^{-1}} \otimes \id_B \right) \, \rho  \right\} \\
&= \mathrm{tr}_{A}\left\{ \left( \sqrt{\sigma_{a|x}^{\mathrm{T}}} \, \left(\sqrt{\sigma_{a|x}^{\mathrm{T}}}\right)^{-1} \, \sigma_{ab|xy}^{\mathrm{T}} \, \left(\sqrt{\sigma_{a|x}^{\mathrm{T}}}\right)^{-1} \sqrt{\sigma_{a|x}^{\mathrm{T}}} \otimes \id_B \right) \sqrt{\sigma_R^{-1}} \, \rho \,  \sqrt{\sigma_R^{-1}}  \right\}  \\
&= \mathrm{tr}_{A}\left\{ \left( \sum_{k,k' \in S_{a,x}} \ket{k}_{a,x}\bra{k}_{a,x} \, \sigma_{ab|xy}^{\mathrm{T}} \, \ket{k'}_{a,x}\bra{k'}_{a,x} \otimes \id_B \right) \sqrt{\sigma_R^{-1}} \, \rho \,  \sqrt{\sigma_R^{-1}}  \right\}  \\
&= \mathrm{tr}_{A}\left\{ \left( \id_{a,x} \, \sigma_{ab|xy}^{\mathrm{T}} \, \id_{a,x} \otimes \id_B \right) \left(\sum_{r,r' \in S_R}  \ket{rr} \bra{r'r'} \right) \right\} \\ \\
&= \mathrm{tr}_{A}\left\{ \left(\sigma_{ab|xy}^{\mathrm{T}}  \otimes \id_B \right) \left(\sum_{r,r' \in S_R}  \ket{rr} \bra{r'r'} \right) \right\} \\ 
&=  \sigma_{ab|xy} \,,
\end{align*}
where we use the fact that Lemma \ref{lem:sup} implies $\id_{a,x} \, \sigma_{ab|xy}^{\mathrm{T}} \, \id_{a,x}  = \sigma_{ab|xy}^{\mathrm{T}}$, and Lemmas \ref{lem:sup} and \ref{lem:tel} imply the last step. 

Now consider the case where $a>0$ and $b=0$:
\begin{align*} \nonumber
&\mathrm{tr}_{A}\left\{ K^\dagger_{a|x} \, \left( M^{ax}_{b|y} + \id_{\overline{ax}} \right) K_{a|x} \otimes \id_B \, \rho  \right\}  \\ 
&= \mathrm{tr}_{A}\left\{ K^\dagger_{a|x} \, M^{ax}_{b|y} \, K_{a|x} \otimes \id_B \, \rho  \right\} + \mathrm{tr}_{A}\left\{ K^\dagger_{a|x} \, \id_{\overline{ax}} \, K_{a|x} \otimes \id_B \, \rho  \right\}  \\ 
&= \mathrm{tr}_{A}\left\{ K^\dagger_{a|x} \, M^{ax}_{b|y} \, K_{a|x} \otimes \id_B \, \rho  \right\} \\
&= \mathrm{tr}_{A}\left\{ \left( \sqrt{\sigma_R^{-1}} \, \sqrt{\sigma_{a|x}^{\mathrm{T}}} \, \left(\sqrt{\sigma_{a|x}^{\mathrm{T}}}\right)^{-1} \, \sigma_{ab|xy}^{\mathrm{T}} \, \left(\sqrt{\sigma_{a|x}^{\mathrm{T}}}\right)^{-1} \sqrt{\sigma_{a|x}^{\mathrm{T}}} \, \sqrt{\sigma_R^{-1}} \otimes \id_B \right) \, \rho  \right\} \\
&= \mathrm{tr}_{A}\left\{ \left( \sqrt{\sigma_{a|x}^{\mathrm{T}}} \, \left(\sqrt{\sigma_{a|x}^{\mathrm{T}}}\right)^{-1} \, \sigma_{ab|xy}^{\mathrm{T}} \, \left(\sqrt{\sigma_{a|x}^{\mathrm{T}}}\right)^{-1} \sqrt{\sigma_{a|x}^{\mathrm{T}}} \otimes \id_B \right) \sqrt{\sigma_R^{-1}} \, \rho \,  \sqrt{\sigma_R^{-1}}  \right\}  \\
&= \mathrm{tr}_{A}\left\{ \left( \sum_{k,k' \in S_{a,x}} \ket{k}_{a,x}\bra{k}_{a,x} \, \sigma_{ab|xy}^{\mathrm{T}} \, \ket{k'}_{a,x}\bra{k'}_{a,x} \otimes \id_B \right) \sqrt{\sigma_R^{-1}} \, \rho \,  \sqrt{\sigma_R^{-1}}  \right\}  \\
&= \mathrm{tr}_{A}\left\{ \left( \id_{a,x} \, \sigma_{ab|xy}^{\mathrm{T}} \, \id_{a,x} \otimes \id_B \right) \left(\sum_{r,r' \in S_R}  \ket{rr} \bra{r'r'} \right) \right\} \\ \\
&= \mathrm{tr}_{A}\left\{ \left(\sigma_{ab|xy}^{\mathrm{T}}  \otimes \id_B \right) \left(\sum_{r,r' \in S_R}  \ket{rr} \bra{r'r'} \right) \right\} \\ 
&=  \sigma_{ab|xy} \,,
\end{align*}
 where we use similar techniques to those in the previous case, together with the fact that Lemma \ref{lem:sup} implies that $K^\dagger_{a|x} \, \id_{\overline{ax}} \, K_{a|x} = 0$. 
 
Consider now the case $a=0$ and $b>0$. Here, the operator inside the trace is: 
 \begin{align*}
 \left( K^\dagger_{a|x} + \id_{\overline{R}} \right) \, M^{ax}_{b|y} \, \left( K_{a|x} + \id_{\overline{R}} \right) &= K^\dagger_{a|x} \, M^{ax}_{b|y} \, K_{a|x} \,,
 \end{align*}
 since Lemma \ref{lem:sup} guarantees that $\id_{\overline{R}} \, M^{ax}_{b|y} = 0 = M^{ax}_{b|y} \, \id_{\overline{R}}$. 
Hence,  
\begin{align*} \nonumber
&\mathrm{tr}_{A}\left\{  \left( K^\dagger_{a|x} + \id_{\overline{R}} \right) \, M^{ax}_{b|y} \, \left( K_{a|x} + \id_{\overline{R}} \right) \otimes \id_B \, \rho  \right\}  \\ 
&= \mathrm{tr}_{A}\left\{ K^\dagger_{a|x} \, M^{ax}_{b|y} \, K_{a|x} \otimes \id_B \, \rho  \right\} \\ 
&= \mathrm{tr}_{A}\left\{ \left( \sqrt{\sigma_R^{-1}} \, \sqrt{\sigma_{a|x}^{\mathrm{T}}} \, \left(\sqrt{\sigma_{a|x}^{\mathrm{T}}}\right)^{-1} \, \sigma_{ab|xy}^{\mathrm{T}} \, \left(\sqrt{\sigma_{a|x}^{\mathrm{T}}}\right)^{-1} \sqrt{\sigma_{a|x}^{\mathrm{T}}} \, \sqrt{\sigma_R^{-1}} \otimes \id_B \right) \, \rho  \right\} \\
&= \mathrm{tr}_{A}\left\{ \left( \sqrt{\sigma_{a|x}^{\mathrm{T}}} \, \left(\sqrt{\sigma_{a|x}^{\mathrm{T}}}\right)^{-1} \, \sigma_{ab|xy}^{\mathrm{T}} \, \left(\sqrt{\sigma_{a|x}^{\mathrm{T}}}\right)^{-1} \sqrt{\sigma_{a|x}^{\mathrm{T}}} \otimes \id_B \right) \sqrt{\sigma_R^{-1}} \, \rho \,  \sqrt{\sigma_R^{-1}}  \right\}  \\
&= \mathrm{tr}_{A}\left\{ \left( \sum_{k,k' \in S_{a,x}} \ket{k}_{a,x}\bra{k}_{a,x} \, \sigma_{ab|xy}^{\mathrm{T}} \, \ket{k'}_{a,x}\bra{k'}_{a,x} \otimes \id_B \right) \sqrt{\sigma_R^{-1}} \, \rho \,  \sqrt{\sigma_R^{-1}}  \right\}  \\
&= \mathrm{tr}_{A}\left\{ \left( \id_{a,x} \, \sigma_{ab|xy}^{\mathrm{T}} \, \id_{a,x} \otimes \id_B \right) \left(\sum_{r,r' \in S_R}  \ket{rr} \bra{r'r'} \right) \right\} \\ \\
&= \mathrm{tr}_{A}\left\{ \left(\sigma_{ab|xy}^{\mathrm{T}}  \otimes \id_B \right) \left(\sum_{r,r' \in S_R}  \ket{rr} \bra{r'r'} \right) \right\} \\ 
&=  \sigma_{ab|xy} \, .
\end{align*}

Finally, consider the case $a=0$ and $b=0$. Here
\begin{align*}
 \left( K^\dagger_{a|x} + \id_{\overline{R}} \right) \,  \left( M^{ax}_{b|y} + \id_{\overline{ax}} \right) \, \left( K_{a|x} + \id_{\overline{R}} \right) &=  K^\dagger_{a|x} \, M^{ax}_{b|y} \, K_{a|x} +  \id_{\overline{R}} \, \id_{\overline{ax}} \,  \id_{\overline{R}}\,,
\end{align*}
where we used the fact that Lemma \ref{lem:sup} implies:
\begin{align*}
 \id_{\overline{R}} \, M^{ax}_{b|y} = 0 = M^{ax}_{b|y} \, \id_{\overline{R}} \quad \textrm{and} \quad K^\dagger_{a|x} \, \id_{\overline{ax}} = 0 = \id_{\overline{ax}} \,  K_{a|x}.
\end{align*}

Hence, 
\begin{align*} \nonumber
&\mathrm{tr}_{A}\left\{  \left( K^\dagger_{a|x} + \id_{\overline{R}} \right)  \,  \left( M^{ax}_{b|y} + \id_{\overline{ax}} \right) \, \left( K_{a|x} + \id_{\overline{R}} \right) \otimes \id_B \, \rho  \right\}  \\ 
&= \mathrm{tr}_{A}\left\{ \left( K^\dagger_{a|x} \, M^{ax}_{b|y} \, K_{a|x} + \id_{\overline{R}} \, \id_{\overline{ax}} \,  \id_{\overline{R}} \right) \otimes \id_B \, \rho  \right\} \\ 
&= \mathrm{tr}_{A}\left\{ K^\dagger_{a|x} \, M^{ax}_{b|y} \, K_{a|x} \otimes \id_B \, \rho  \right\} + \mathrm{tr}_{A}\left\{  \id_{\overline{R}} \, \id_{\overline{ax}} \,  \id_{\overline{R}} \otimes \id_B \, \rho  \right\} \\ 
&= \mathrm{tr}_{A}\left\{ K^\dagger_{a|x} \, M^{ax}_{b|y} \, K_{a|x} \otimes \id_B \, \rho  \right\} +  \mathrm{tr}_{A}\left\{  \id_{\overline{ax}} \otimes \id_B \, (\id_{\overline{R}} \otimes \id_B) \rho (\id_{\overline{R}} \otimes \id_B)  \right\} \\ 
&= \mathrm{tr}_{A}\left\{ K^\dagger_{a|x} \, M^{ax}_{b|y} \, K_{a|x} \otimes \id_B \, \rho  \right\} \\ 
&= \mathrm{tr}_{A}\left\{ \left( \sqrt{\sigma_R^{-1}} \, \sqrt{\sigma_{a|x}^{\mathrm{T}}} \, \left(\sqrt{\sigma_{a|x}^{\mathrm{T}}}\right)^{-1} \, \sigma_{ab|xy}^{\mathrm{T}} \, \left(\sqrt{\sigma_{a|x}^{\mathrm{T}}}\right)^{-1} \sqrt{\sigma_{a|x}^{\mathrm{T}}} \, \sqrt{\sigma_R^{-1}} \otimes \id_B \right) \, \rho  \right\} \\
&= \mathrm{tr}_{A}\left\{ \left( \sqrt{\sigma_{a|x}^{\mathrm{T}}} \, \left(\sqrt{\sigma_{a|x}^{\mathrm{T}}}\right)^{-1} \, \sigma_{ab|xy}^{\mathrm{T}} \, \left(\sqrt{\sigma_{a|x}^{\mathrm{T}}}\right)^{-1} \sqrt{\sigma_{a|x}^{\mathrm{T}}} \otimes \id_B \right) \sqrt{\sigma_R^{-1}} \, \rho \,  \sqrt{\sigma_R^{-1}}  \right\}  \\
&= \mathrm{tr}_{A}\left\{ \left( \sum_{k,k' \in S_{a,x}} \ket{k}_{a,x}\bra{k}_{a,x} \, \sigma_{ab|xy}^{\mathrm{T}} \, \ket{k'}_{a,x}\bra{k'}_{a,x} \otimes \id_B \right) \sqrt{\sigma_R^{-1}} \, \rho \,  \sqrt{\sigma_R^{-1}}  \right\}  \\
&= \mathrm{tr}_{A}\left\{ \left( \id_{a,x} \, \sigma_{ab|xy}^{\mathrm{T}} \, \id_{a,x} \otimes \id_B \right) \left(\sum_{r,r' \in S_R}  \ket{rr} \bra{r'r'} \right) \right\} \\ \\
&= \mathrm{tr}_{A}\left\{ \left(\sigma_{ab|xy}^{\mathrm{T}}  \otimes \id_B \right) \left(\sum_{r,r' \in S_R}  \ket{rr} \bra{r'r'} \right) \right\} \\ 
&=  \sigma_{ab|xy} \,,
\end{align*} 
where we used similar techniques to those in the previous cases, together with the fact that $\id_{\overline{R}} \otimes \id_B \ket{\psi} = 0$.

We see then that the state defined in Eq.~\eqref{eq:1}, together with the measurement operators from Eqs.~\eqref{eq:2} and \eqref{eq:3} are indeed well defined and reproduce the assemblage. Hence, it follows that the GHJW theorem generalises to these sequential scenarios. $\square$

\section{Steering inequality as a certificate of post-quantumness}

In this section, we present a method to certify that an assemblage is post-quantum. The method relies on defining a steering inequality, and showing that the candidate post-quantum assemblage can violate this inequality -- i.e., achieve a value that no quantum assemblages can reach. This method does not necessary always succeed at detecting post-quantum assemblages, since it is the first step in a more complex technique developed in a forthcoming work \cite{ourhier}.

A general linear steering functional in the Bob-with-input scenario is of the form:
\begin{align}\label{eq:genineq}
 \beta(\{\sigma_{a|xy}\}) = \Tr\left\{ \sum_{a,x,y} F_{axy} \, \sigma_{a|xy}   \right\}\,,
\end{align}
where $\{F_{axy}\}$ is a set of Hermitian operators.

\subsection{Local Hidden State models}

Traditionally, a steering inequality is used to test whether an assemblage admits a classical model. Indeed, when the value to which an assemblage evaluates the functional of Eq.~\eqref{eq:genineq} goes below the smallest value (or above the largest value) that can be obtained with ``unsteerable'' assemblages, then the assemblage is said to ``violate'' the inequality. Therefore, a key ingredient in constructing a steering inequality from the steering functional of Eq.~\eqref{eq:genineq} is to define the set of ``unsteerable'', a.k.a.~\textit{Local Hidden States} (LHS), assemblages. An LHS assemblage is one that may be prepared by Alice and Bob when their actions can only be correlated via a classical random variable, i.e., when they share a source of classical randomness. As discussed in Ref.~\cite{Schmid19}, an LHS model is one that can be freely prepared by using \textit{Local Operations and Shared Randomness} (a.k.a.~LOSR operations). In the context of the Bob-with-input scenario, then, unsteerable assemblages correspond to the following: 
\begin{defn}\textbf{LHS assemblages in the Bob-with-input scenario.}\\
An assemblage $\{\sigma_{a|xy}\}_{a,x,y}$ in the Bob-with-input scenario admits a Local Hidden State (LHS) model, if there exists a random variable $\lambda \in \Lambda$, a set $\{\rho_{\lambda,y}\}_{\lambda,y}$ of normalised states for Bob, and a conditional probability distribution $p(a|x,\lambda)$, such that:
\begin{align*}
\sigma_{a|xy} = \sum_\lambda \, p(\lambda) \, p(a|x,\lambda) \, \rho_{\lambda,y}\,,
\end{align*}
where $p(\lambda)$ is a normalised probability distribution over the variable $\lambda$.
\end{defn}
Notice that, similar to the traditional steering scenario \cite{PaulSDP}, checking whether an assemblage admits an LHS model is hence an instance of a semi-definite program (SDP). 

Given a steering functional as per Eq.~\eqref{eq:genineq}, then a steering inequality is defined as:
\begin{align}\label{eq:genineqbound}
\beta(\{\sigma_{a|xy}\}) \geq \beta^{\mathrm{LHS}}\,,
\end{align}
where 
\begin{align}
\beta^{\mathrm{LHS}} := \min_{\{\sigma_{a|xy}\} \, \text{is LHS}} \beta(\{\sigma_{a|xy}\})\,.
\end{align}
Therefore, should an assemblage $\{\sigma^\prime_{a|xy}\}$ yield $\beta(\{\sigma^\prime_{a|xy}\}) < \beta^{\mathrm{LHS}}$, the assemblage would be steerable. 

\subsection{Quantum bounds: necessity for numerical estimation}

Notice that up to here we have described a steering inequality as a test to certify whether an assemblage is steerable or not, rather than whether it is post-quantum or not. Should one wish to certify post-quantumness of an assemblage, then one should focus on the quantum-steering inequality:

\begin{align}
\beta^{\mathrm{Q}} \leq \Tr\left\{ \sum_{a,x,y} F_{axy} \, \sigma_{a|xy}   \right\}\,,
\end{align}
where 
\begin{align}
\beta^{\mathrm{Q}} := \min_{\{\sigma_{a|xy}\} \, \text{is quantum}} \Tr\left\{ \sum_{a,x,y} F_{axy} \, \sigma_{a|xy}   \right\}\,.
\end{align}

Computing the minimum quantum value of a steering or Bell functional is known to be a complex computational task, which yields a definite value only for certain inequalities. In the case of Bell scenarios, there are numerical tools to lower bound the value of $\beta^{\mathrm{Q}}$ \cite{NPA}, which may be adapted to traditional steering scenarios \cite{CavSkrAgu+15,pqs}. In the scenario we are considering here, however, no such computational tools exist in the literature. In the following, we will discuss how one may develop a technique, similar in spirit to the so-called 1+AB level of the  Navascu\'es-Pironio-Ac\'in (NPA) hierarchy  \cite{NPA}, which will allow us to lower (or upper) bound the value of $\beta^{\mathrm{Q}} $ for any steering inequality, in particular the steering inequality relevant to this manuscript. 

\subsection{$\aQ$ bounds: numerical lower bound for $\beta^{\mathrm{Q}}$}

In this section we present the details of a relaxation of the set of quantum assemblages in the Bob-with-Input scenario, which we call $\aQ$ assemblages. 
This relaxation will allow us to compute lower bounds on the quantum bound of any steering functional in the Bob-with-Input scenario.

For simplicity in the presentation, let us consider the case where Alice performs dichotomic measurements. The case for arbitrary number of outcomes will be properly discussed within the more general technique of our upcoming work \cite{ourhier}. Similarly to the almost-quantum assemblages of Ref.~\cite{pqs}, consider a moment matrix $\Gamma$ whose rows and columns are labelled by the `words' from the following set: 
\begin{align}
\mathcal{S} := \{\emptyset\} \cup \{ x \}_{x=1:m_A} \cup \{ y \}_{y=1:m_B} \cup \{ xy \}_{x=1:m_A\,,\, y=1:m_B} \,,
\end{align}
where $m_A$ ($m_B$) is the cardinality of the set of inputs of Alice (Bob). The elements of $\Gamma$ are taken to be square complex matrices of dimension given by that of Bob's Hilbert space.

The moment matrix $\Gamma$ is asked to fulfill the following properties: 
\begin{defn} \textbf{Moment matrix} $\boldsymbol{\Gamma}$ .-- \label{def:MM}\\
A matrix $\Gamma$ (whose rows and columns are labelled by the words in the set $\mathcal{S}$) is a moment matrix if it satisfies the following properties: 
\begin{align}
& \Gamma \geq 0 \,,\\
& \Gamma(\emptyset,\emptyset) = \id_d \label{c8-2}\,,\\
& \Gamma(v,v) = \Gamma(\emptyset,v) \quad \forall \, v \in \mathcal{S} \label{c8-3}\,,\\
& \Gamma(\emptyset,xy) =  \Gamma(x,xy) = \Gamma(y,xy) = \Gamma(x,y) \quad \forall \, x,y \label{c8-4} \,,\\
& \Gamma(x,x'y) = \Gamma(xy,x'y) = \Gamma(xy,x') \quad \forall \, x,x',y \label{c8-5} \,,\\
& \Gamma(y,xy') = \Gamma(xy,xy') = \Gamma(xy,y') \quad \forall \, x,y,y' \label{c8-6}\,,\\
& \Gamma(x,x') \propto \id_d \quad \forall \, x,x' \label{c8-7}\,,
\end{align}
where $d$ is the dimension of Bob's Hilbert space. 
\end{defn}

We can now define the set of $\aQ$ assemblages:

\begin{defn} \textbf{Set of $\boldsymbol{\aQ}$ assemblages.--} \label{def:AQ}\\
An assemblage $\{\sigma_{a|xy}\}$ in the Bob-with-Input scenario is a ${\aQ}$ assemblage iff there exists a moment matrix $\Gamma$ such that: 
\begin{align}
& \Gamma(\emptyset,x) = \Tr\left( \sigma_{0|xy} \right) \, \id_d \quad \forall \, x \label{c9-1}\,,\\
& \Gamma(\emptyset,y) = \frac{1}{d} \, \sigma^\mathrm{T}_y \quad \forall \, y \label{c9-2}\,,\\
& \Gamma(\emptyset, xy) = \frac{1}{d} \, \sigma^\mathrm{T}_{0|xy} \quad \forall \, x,y \label{c9-3}\,,
\end{align}
where $\sigma_y = \sum_a \sigma_{a|xy}$ is the marginal state on Bob's side, and $^\mathrm{T}$ denotes transposition.
\end{defn}

Notice that testing whether an assemblage belongs to $\aQ$ is hence a single instance of a semidefinite program (SDP). These can be efficiently computed via standard computational tools \cite{matlab, cvx, cvx2, sdpt3}. 

The minimum value that $\aQ$ assemblages can yield for a given steering inequality is then: 
\begin{align}
\beta^{\aQ} := \min_{\{\sigma_{a|xy}\} \in \aQ} \Tr\left\{ \sum_{a,x,y} F_{axy} \, \sigma_{a|xy}   \right\}\,,
\end{align}
and computing $\beta^{\aQ}$ requires a single instance of an SDP. 

We will now show that the set of quantum assemblages is included within $\aQ$, from which follows that $\beta^{\aQ} \leq \beta^{\mathrm{Q}}$. 

\begin{thm}
An assemblage $\{\sigma_{a|xy}\}$ in the Bob-with-Input scenario is a ${\aQ}$ assemblage if it is a quantum assemblage. 
\end{thm}
\begin{proof}
Let us start from a quantum assemblage $\{\sigma_{a|xy}\}$, and show that it belongs to $\aQ$. For this, we need to show that there exist a moment matrix $\Gamma$ that satisfies the conditions of Def.~\ref{def:AQ}. 

Let the state $\rho$,  the collection of projective measurements $\{\Pi_{a|x}\}$, and the CPTP maps $\{\mathcal{E}_y\}$ provide a quantum realisation of $\{\sigma_{a|xy}\}$. Notice that, due to  Stinespring dilation \cite{vp,stine}, we can always take Alice's measurements to be projective and the state $\rho$ to be pure, without loss of generality. Denote by $B_{aux}$ the auxiliary system in Bob's lab which is necessary to dilate the CPTP maps $\{\mathcal{E}_y\}$ into unitaries $U_y$, and denote by $\ket{\psi}$ the quantum state on the Hilbert space $\cH_A \otimes \cH_B \otimes \cH_{B_{aux}}$ that satisfies: 
\begin{align*}
\sigma_{a|xy} = \Tr_{A} \left( \Pi_{a|x} \otimes \mathcal{E}_y \, [\rho] \right) = \Tr_{AB_{aux}} \left( \Pi_{a|x} \otimes U_y \, [\ket{\psi} \bra{\psi}] \right)\,.
\end{align*}

Construct now the following state and measurements: 
\begin{align}
&\tilde{\rho} := \ket{\psi} \bra{\psi} \otimes \id_d^{B'} \,,\\
&\Sigma^{BB'B_{aux}}_{0|y} = \left(U^\dagger_y \otimes \id_{B'} \right) \left(\id_{B_{aux}} \otimes \ket{\Phi}\bra{\Phi}\right) \left(U_y \otimes \id_{B'} \right)\,,\\
&\Sigma^{BB'B_{aux}}_{1|y}= \id_{BB'B_{aux}} -\Sigma^{BB'B_{aux}}_{0|y}\,,
\end{align}
where:
\begin{compactitem}
\item $d$ is the dimension of the Hilbert space of $\{\sigma_{a|xy}\}$,
\item $B'$ is an auxiliary system of dimension $d$ in Bob's lab,
\item $\ket{\Phi}$ is a maximally entangled state of two qudits.
\end{compactitem}
Notice that the operator $\Sigma^{BB'B_{aux}}_{0|y}$ is a projector: on the one hand, $\Sigma^{BB'B_{aux}}_{0|y}$ is a unitary applied to a state, and since unitaries preserve hermiticity then $\Sigma^{BB'B_{aux}}_{0|y}$ is hermitian; on the other hand, 
\begin{align*}
\left(\Sigma^{BB'B_{aux}}_{0|y}\right)^\dagger \Sigma^{BB'B_{aux}}_{0|y} &= \left(U^\dagger_y \otimes \id_{B'} \right) \left(\id_{B_{aux}} \otimes \ket{\Phi}\bra{\Phi}\right) \left(U_y \otimes \id_{B'} \right)\left(U^\dagger_y \otimes \id_{B'} \right) \left(\id_{B_{aux}} \otimes \ket{\Phi}\bra{\Phi}\right) \left(U_y \otimes \id_{B'} \right) \\
&= \left(U^\dagger_y \otimes \id_{B'} \right) \left(\id_{B_{aux}} \otimes \ket{\Phi}\bra{\Phi}\right) \left(\id_{B_{aux}} \otimes \ket{\Phi}\bra{\Phi}\right) \left(U_y \otimes \id_{B'} \right) \\
&= \left(U^\dagger_y \otimes \id_{B'} \right) \left(\id_{B_{aux}} \otimes \ket{\Phi}\bra{\Phi}\right) \left(U_y \otimes \id_{B'} \right) \\
&=  \Sigma^{BB'B_{aux}}_{0|y} \,.
\end{align*}

Now define the following operators: 
\begin{align}
&\mathbb{O}_{x} = \Pi_{0|x} \otimes \id_{BB'B_{aux}} \,,\\
&\mathbb{O}_{y} = \id_A \otimes \Sigma^{BB'B_{aux}}_{0|y} \,,\\
&\mathbb{O}_{xy} =  \Pi_{0|x} \otimes \Sigma^{BB'B_{aux}}_{0|y} \,,\\
&\mathbb{O}_{\emptyset} = \id_{ABB'B_{aux}}\,.
\end{align}
Notice that, by definition, $\mathbb{O}_v$ is a projector $\forall \, v \in \mathcal{S}$.

Now we can define a matrix $\Gamma$ as follows: 
\begin{align}
\Gamma(u,v) := \Tr_{ABB_{aux}} \left( \mathbb{O}^\dagger_u \, \mathbb{O}_v \, \tilde{\rho} \right)\,,
\end{align}
where $u,v \in \mathcal{S}$. Notice that the elements $\Gamma(u,v)$ are d-dimensional operators in the Hilbert space $\cH_{B'}$. 

We will now show that $\Gamma$ is a moment matrix that satisfies the constraints of Definitions \ref{def:MM} and \ref{def:AQ}. 
\begin{itemize}
\item $\Gamma(\emptyset,\emptyset) = \Tr_{ABB_{aux}} \left( \tilde{\rho}\right) = \Tr_{ABB_{aux}} \left( \ket{\psi} \bra{\psi}\right) \id_d^{B'} = \id_d$, hence Eq.~\eqref{c8-2} is satisfied. 
\item Since the operators $\mathbb{O}_v$ are projectors, then $\Gamma(v,v) = \Tr_{ABB_{aux}} \left( \mathbb{O}^\dagger_v \, \mathbb{O}_v \, \tilde{\rho} \right) = \Tr_{ABB_{aux}} \left( \mathbb{O}_v \, \tilde{\rho} \right) = \Gamma(\emptyset,v)$, hence Eq.~\eqref{c8-3} is satisfied. 
\item The proofs that equations \eqref{c8-4}, \eqref{c8-5}, and \eqref{c8-6} are satisfied, follow similarly from the fact that the operators $\mathbb{O}_v$, $\Pi_{0|x}$, and $ \Sigma^{BB'B_{aux}}_{0|y}$ are projectors. As an example, let us present the case of Eq.~\eqref{c8-5}. First, notice that the properties of the before-mentioned projection operators imply that:
\begin{align*}
\Tr_{ABB_{aux}} & \left( \left( \Pi_{0|x} \otimes \id_{BB'B_{aux}}\right)^\dagger \, \Pi_{0|x'} \otimes \Sigma^{BB'B_{aux}}_{0|y} \, \tilde{\rho} \right) \\
&= \Tr_{ABB_{aux}} \left( \left( \Pi_{0|x} \otimes \id_{BB'B_{aux}}\right)^\dagger \, \Pi_{0|x'} \otimes (\Sigma^{BB'B_{aux}}_{0|y})^\dagger \Sigma^{BB'B_{aux}}_{0|y} \, \tilde{\rho} \right) \\
&= \Tr_{ABB_{aux}} \left( \left( \Pi_{0|x} \otimes \Sigma^{BB'B_{aux}}_{0|y} \right)^\dagger \, \Pi_{0|x'} \otimes \Sigma^{BB'B_{aux}}_{0|y} \, \tilde{\rho} \right)\,,
\end{align*}
from which $\Gamma(x,x'y) = \Gamma(xy,x'y)$ follows, and
\begin{align*}
\Tr_{ABB_{aux}} \left( \left( \Pi_{0|x} \otimes \id_{BB'B_{aux}}\right)^\dagger \, \Pi_{0|x'} \otimes \Sigma^{BB'B_{aux}}_{0|y} \, \tilde{\rho} \right) &= 
\Tr_{ABB_{aux}} \left( \left( \Pi_{0|x} \otimes \id_{BB'B_{aux}}\right)^\dagger \, \Pi_{0|x'} \otimes (\Sigma^{BB'B_{aux}}_{0|y})^\dagger \, \tilde{\rho} \right) \\
&= \Tr_{ABB_{aux}} \left( \left( \Pi_{0|x} \otimes \Sigma^{BB'B_{aux}}_{0|y} \right)^\dagger \, \Pi_{0|x'} \otimes \id_{BB'B_{aux}} \, \tilde{\rho} \right),
\end{align*}
from which $\Gamma(x,x'y) = \Gamma(xy,x')$ follows. Hence, Eq.~\eqref{c8-5} is satisfied.
\item To see that conditions given by Eq.~\eqref{c8-7} are satisfied, notice that: 
\begin{align*}
\Tr_{ABB_{aux}} & \left( \left( \Pi_{0|x} \otimes \id_{BB'B_{aux}}\right)^\dagger \, \Pi_{0|x'} \otimes  \id_{BB'B_{aux}} \, \tilde{\rho} \right) \\
&= \Tr_{ABB_{aux}} \left( \left( \Pi_{0|x} \otimes \id_{BB'B_{aux}}\right)^\dagger \, \Pi_{0|x'} \otimes  \id_{BB'B_{aux}} \, \ket{\psi} \bra{\psi} \otimes \id_d^{B'}  \right) \\
&= \Tr_{ABB_{aux}} \left( \left( \Pi_{0|x} \otimes \id_{BB_{aux}}\right)^\dagger \, \Pi_{0|x'} \otimes  \id_{BB_{aux}} \, \ket{\psi} \bra{\psi} \right) \, \id_d^{B'}  \\
&\propto \id_d^{B'} \,,
\end{align*}
from which $\Gamma(x,x')\propto \id_d$ follows. 
\item For the first condition of Def.~\ref{def:AQ}, notice that: 
\begin{align*}
\Gamma(\emptyset,x) &= \Tr_{ABB_{aux}}\left( \Pi_{0|x} \otimes  \id_{BB'B_{aux}} \, \tilde{\rho} \right) = \Tr_{ABB_{aux}}\left( \Pi_{0|x} \otimes  \id_{BB'B_{aux}} \,\ket{\psi} \bra{\psi} \otimes \id_d^{B'} \right) \\
&= \Tr_{ABB_{aux}}\left( \Pi_{0|x} \otimes  \id_{BB_{aux}} \,\ket{\psi} \bra{\psi} \right) \id_d^{B'} = \Tr_{ABB_{aux}}\left( \Pi_{0|x} \otimes  U_y \,[\ket{\psi} \bra{\psi}] \right) \id_d^{B'} \\
&= \Tr_{B}\left( \sigma_{0|xy} \right) \id_d \,,
\end{align*}
where the equality in the second line follows from the fact that unitary operations are trace preserving. Hence, the conditions given by Eq.~\eqref{c9-1} are satisfied.
\item For the third condition of Def.~\ref{def:AQ}, notice that: 
\begin{align*}
\Gamma(\emptyset,xy) &=  \Tr_{ABB_{aux}}\left( \Pi_{0|x} \otimes \Sigma^{BB'B_{aux}}_{0|y}\, \tilde{\rho} \right) \\
&= \Tr_{ABB_{aux}}\left( \Pi_{0|x} \otimes \left( \left(U^\dagger_y \otimes \id_{B'} \right) \left(\id_{B_{aux}} \otimes \ket{\Phi}\bra{\Phi}\right) \left(U_y \otimes \id_{B'} \right) \right) \, \tilde{\rho} \right) \\
&= \Tr_{ABB_{aux}} \left( \left(\Pi_{0|x}^\dagger \otimes U^\dagger_y \otimes \id_{B'} \right) \left(\id_{AB_{aux}} \otimes \ket{\Phi}\bra{\Phi}\right) \left(\Pi_{0|x} \otimes U_y \otimes \id_{B'} \right) \, \tilde{\rho} \right) \\
&= \Tr_{ABB_{aux}} \left( \left(\id_{AB_{aux}} \otimes \ket{\Phi}\bra{\Phi}\right) \left( \left(\Pi_{0|x} \otimes U_y \otimes \id_{B'} \right) \, \tilde{\rho} \, \left(\Pi_{0|x}^\dagger \otimes U^\dagger_y \otimes \id_{B'} \right) \right) \right) \\
&= \Tr_{ABB_{aux}} \left( \left(\id_{AB_{aux}} \otimes \ket{\Phi}\bra{\Phi}\right) \left( \left(\Pi_{0|x} \otimes U_y \right) \, \ket{\psi} \bra{\psi} \, \left(\Pi_{0|x}^\dagger \otimes U^\dagger_y \right) \otimes \id_{B'} \right) \right) \\
&= \Tr_{B} \left( \ket{\Phi}\bra{\Phi} \sigma_{0|xy} \otimes \id_{B'} \right) \\
&= \frac{1}{d} \, \sigma^\mathrm{T}_{0|xy} \,.
\end{align*}
Hence, the conditions from Eq.~\eqref{c9-3} are satisfied. 
\item The second condition of Def.~\ref{def:AQ}, follows similarly from the third one plus the fact that Alice's measurements $\{\Pi_{a|x}\}$ are complete:
\begin{align*}
\Gamma(\emptyset,y) &= \Tr_{ABB_{aux}}\left( \id_A \otimes \Sigma^{BB'B_{aux}}_{0|y}\, \tilde{\rho} \right) = \sum_{a=0,1} \Tr_{ABB_{aux}}\left( \Pi_{a|x} \otimes \Sigma^{BB'B_{aux}}_{0|y}\, \tilde{\rho} \right) = \sum_{a=0,1}\frac{1}{d} \, \sigma^\mathrm{T}_{a|xy} = \frac{1}{d} \, \sigma^\mathrm{T}_{y} \,,
\end{align*}
where the third equality can be proven following the same steps as in the proof for the third condition of Def.~\ref{def:AQ}, but setting $a=1$. It hence follows that the matrix $\Gamma$ satisfies the conditions of Eq.~\eqref{c9-2}. 
\item The last condition that needs to be proven is that $\Gamma$ is positive semidefinite. For this, first notice that the $(i,j)$ entry of the element $\Gamma(u,v)$ is given by:
\begin{align*}
\Gamma^{i,j}(u,v) &= \bra{i} \Tr_{ABB_{aux}} \left( \mathbb{O}^\dagger_u \, \mathbb{O}_v \, \tilde{\rho} \right) \ket{j} 
= \bra{i} \Tr_{ABB_{aux}} \left( \mathbb{O}^\dagger_u \, \mathbb{O}_v \, \ket{\psi} \bra{\psi} \otimes \id_d^{B'} \right) \ket{j} 
= \bra{i} \bra{\psi} \mathbb{O}^\dagger_u \, \mathbb{O}_v \ket{\psi} \ket{j}\,.
\end{align*}
Now define the vectors $\ket{\phi_{i,v}} = \mathbb{O}_v \ket{\psi} \ket{i}$. Hence, $\Gamma^{i,j}(u,v) = \langle \phi_{i,u} \lvert \phi_{j,v} \rangle$. This shows that $\Gamma$ is a Gramian matrix, and therefore positive semidefinite. 
\end{itemize}
We see then that, for every quantum assemblage $\{\sigma_{a|xy}\}$, one can construct a moment matrix $\Gamma$ which (i) satisfies the conditions of Def.~\ref{def:MM}, i.e., its a valid moment matrix, and (ii) satisfies the conditions of Def.~\ref{def:AQ}. Therefore, quantum assemblage belong to the set $\aQ$. 
\end{proof}

\section{Proof of {Theorem \ref{thm:ex}:} {post-quantum steering does not imply post-quantum non-locality in the Bob-with-input scenario}}

In this section we will prove {Theorem \ref{thm:ex}}, which says that post-quantum steering in the Bob-with-input setting does not imply post-quantum non-locality. To prove this we introduce a new set of Bob-with-input assemblages, which strictly contain the quantum Bob-with-input assemblages, but provably can never give rise to post-quantum correlations. We will call this set the set of \textit{PTP assemblages}, since {they may be mathematically expressed by means of positive} and trace-preserving (PTP) maps, which may not necessarily be completely positive. We will define these assemblages, but first it is instructive to {recall} the definition of quantum Bob-with-input assemblages.
\\
\\
\textbf{{Definition \ref{def:quant_gen}}.} \textit{\textbf{Quantum Bob-with-input assemblages.}}\\
\textit{An assemblage $\{\sigma_{a|xy}\}_{a,x,y}$ has a has a quantum realisation in the steering scenario where Bob has an input if and only if there exists a Hilbert space $\cH_A$ and POVMs $\{ M_{a|x} \}_{a,x}$ for Alice,  a state $\rho$ in $\cH_A \otimes \cH_B$, and a collection of CPTP maps $\{\mathcal{E}_y\}_y$ in $\cH_B$ for Bob, such that}
\begin{equation}
\sigma_{a|xy} = \mathcal{E}_y \left[ \mathrm{tr}_{A}\left\{ (M_{a|x} \otimes \id  ) \rho   \right\} \right].
\end{equation}

A relaxation of the quantum assemblages from the previous setup is that where $\mathcal{E}^y$ is a PTP map (but not necessarily completely positive) instead:
\begin{defn}\textbf{PTP assemblage.} \\
An assemblage $\{\sigma_{a|xy}\}_{a,x,y}$ is a PTP assemblage iff there exists a Hilbert space $\cH_A$ for Alice, POVMs $\{ M_{a|x} \}_{a,x}$ for Alice, a state $\rho$ in $\cH_A \otimes \cH_B$, and a collection of PTP maps $\{\mathcal{E}^y\}_y$ in $\cH_B$ for Bob, such that $\sigma_{a|xy} = \mathcal{E}^y \left[ \mathrm{tr}_{A}\left\{ M_{a|x} \otimes \id  \, \rho   \right\} \right]$.
\end{defn} 

{When an assemblage is mathematically expressed via some PTP maps that are not completely positive, then there is a possibility for it to be post-quantum.} This method has been previously used in Ref.~\cite{oform} to construct post-quantum assemblages that may only exhibit quantum correlations in Bell scenarios. {We remark that, even if an assemblage is a PTP assemblage, this does not mean that it must be physically prepared in the laboratory by allowing Bob to physically implement a possibly not completely positive map. This is just a convenient mathematical characterisation of a set of assemblages, with no underlying physical connotations. }

In the bipartite steering scenarios considered here, this method {for mathematically constructing} assemblages has a similar advantage to that in \cite{oform}: when Bob measures on his system, the correlations $p(ab|xy)$ are compatible with quantum theory. This claim is formalised below.

\begin{thm}\label{PTPquantum}
{Consider a Bob-with-input scenario, and define the assemblage $\sigma_{a|xy} = \mathrm{tr}_{A}\left\{ M_{a|x} \otimes \mathcal{E}^y \, \rho  \right\}$, where $\rho$ is the state of a quantum system shared by Alice and Bob, $\{M_{a|x}\}_a$ a POVM for Alice for each $x$, and $\mathcal{E}^y$ a PTP map for Bob, for each $y$.}

Let $\{N_b\}_b$ be a POVM for Bob. Then, the correlations $p(ab|xy) = \Tr\left\{ N_b \, \sigma_{a|xy} \right\}$ 
have a quantum realisation in a Bell experiment. Moreover, let $\{N_{b|z}\}_b$ be a POVM for Bob, for each $z$. Then, the correlations 
$ p(ab|xyz) = \Tr\left\{ N_{b|z} \, \sigma_{a|xy} \right\}$
also have a quantum realisation in a Bell experiment.
\end{thm}

\begin{proof}
Let's start from the correlations $p(ab|xy)=\mathrm{tr}\left\{N_{b} \, \sigma_{a|x,y}\right\}$. For each PTP map $\mathcal{E}^y$, there exists the dual map $\mathcal{F}^y$, which is positive and unital, such that
\begin{align*}
p(ab|xy) &= \mathrm{tr}\left\{ M_{a|x} \otimes N_{b} \, \left( \id_A \otimes \mathcal{E}^{y} \right) [\rho]  \right\}\\
&=\mathrm{tr}\left\{ M_{a|x} \otimes \mathcal{F}^y(N_{b}) \, \rho  \right\}.
\end{align*}
For each $y$, since the map $\mathcal{F}^y$ is positive and unital, it takes a general measurement $N_{b}$ to a general measurement with POVM elements $M_{b|y}:=\mathcal{F}^y(N_{b})$. Therefore,
\begin{align*}
p(ab|xy) = \Tr\left\{ M_{a|x} \otimes M_{b|y} \, \rho \right\}\,,
\end{align*}
which gives a quantum realisation for the correlations $p(ab|xy)$.

Now let's move on to the second statement in the theorem. Let Bob choose from a set of general measurements indexed by the variable $z$ with POVM elements $N_{b|z}$. In this case the correlations are of the form 
\begin{align*}
p(ab|xyz) &= \mathrm{tr}\left\{ M_{a|x} \otimes N_{b|z} \, \left( \id_A \otimes \mathcal{E}^{y} \right) [\rho]  \right\}\\
&=\mathrm{tr}\left\{ M_{a|x} \otimes \mathcal{F}^y(N_{b|z}) \, \rho  \right\}.
\end{align*}
By a similar argument as before, $M_{b|yz} := \mathcal{F}^y(N_{b|z})$ defines a POVM for each choice of $(y,z)$. By reinterpreting $(y,z)$ as Bob's total input, this provides a quantum realisation for the correlations $p(ab|xyz)$.
\end{proof}

\

Theorem \ref{PTPquantum} shows that PTP assemblages cannot generate post-quantum correlations in a Bell scenario. 
In order to certify the post-quantumness of a PTP assemblage, thus, we need a new method that does not rely on Bell correlations. {Here we will use the method based on the steering inequality violation, discussed in the previous section. }

Now we are in a position to present the proof of {Theorem \ref{thm:ex}}, {recalled} below for convenience. 
\\
\\
\textbf{{Theorem \ref{thm:ex}}.} \textit{The following PTP assemblage has no quantum realisation:}
{
\begin{align*}
&\sigma^*_{a|10} = \tfrac{1}{2} (\ket{0}+(-1)^a \ket{1})(\bra{0}+(-1)^a \bra{1})\,, \\
&\sigma^*_{a|20} = \tfrac{1}{2} (\ket{0}-(-1)^a \, i \,  \ket{1})(\bra{0}+(-1)^a \, i \,  \bra{1})\,, \\
&\sigma^*_{0|30} = \tfrac{1}{2} \ket{0}\bra{0}\,, \quad \sigma^*_{1|30} = \tfrac{1}{2} \ket{1}\bra{1}\,, \\
&\sigma^*_{a|x1} = \sigma_{a|x0}^{* \mathrm{T}}\,. 
\end{align*}
}
\begin{proof}
First notice that the assemblage $\{\sigma^*_{a|xy}\}$, where {$x \in \{1,2,3\}$,} $a \in \{0,1\}$ and $y \in \{0,1\}$,  may {be mathematically expressed as one where} Alice and Bob share a maximally entangled state of two qubits, Alice chooses between the three Pauli measurements to perform on her qubit, and Bob applies on his qubit the identity channel when $y=0$ and the transpose when $y=1$. {This convenient mathematical representation shows that} $\{\sigma^*_{a|xy}\}$ is indeed a PTP assemblage, and by virtue of Theorem \ref{PTPquantum} may never yield post-quantum correlations in Bell experiments. 

{In order to show that the assemblage is post-quantum, let us consider the steering functional, which we will denote by $\beta^*$,  given by the following operators: 
\begin{align}\label{eq:theops}
F^*_{axy} = \frac{1}{2}(\openone - (-1)^a \sigma_x )^{\mathrm{T}^y}\,.
\end{align}
That is, 
\begin{align}\label{theineq}
 \beta^*(\{\sigma_{a|xy}\})= \Tr\left\{ \sum_{a,x,y} F^*_{axy} \, \sigma_{a|xy}   \right\}\,.
\end{align}
For the assemblage defined in Eq.~\eqref{e:example}, it follows that $\beta^* = 0$. In addition, the minimum value achieved by Non-signalling Bob-with-input assemblages is $\beta^{\mathrm{NS}} = 0$, and the LHS bound of the inequality is $\beta^{\mathrm{LHS}}=1.2679$. 

To lower bound the minimum quantum value of the steering functional of Eq.~\eqref{theineq}, we compute the minimum value achieved by the set $\aQ$ of assemblages. The solution to the SDP (computed in Matlab \cite{matlab}, using the software CVX \cite{cvx,cvx2} and the solver SDPT3 \cite{sdpt3})  yields $\beta^{\aQ}=0.4135$, which implies that $\beta^{\mathrm{Q}} \geq 0.4135 > 0$. Hence, the assemblage of Eq.~\eqref{e:example} is post-quantum and certified by this quantum-steering inequality. }
\end{proof}

An alternative proof that the assemblage of Eq.~\eqref{e:example} is post-quantum may also be found analytically. 
For this, we develop a method that can certify the post-quantumness of a PTP assemblage without relying on Bell correlations, and is based on the following observation:

\begin{lemma}\label{theIDthing}
{Consider the Bob-with-input scenario.} Let $\{\sigma_{a|xy}\}$ be a quantumly realisable assemblage, with the following property: $\sigma_{a|x1}$ is (proportional to) a pure state for all $a,x$. Then, there exist CPTP maps $\{\mathcal{F}^y\}_{y>1}$ such that $\sigma_{a|xy} = \mathcal{F}^y[\sigma_{a|x1}] \quad \forall \, y>1$.
\end{lemma}

In order to prove this Lemma, we first need to show a more general result, stated as a theorem below. The proofs will be presented in the diagrammatic notation of Ref.~\cite{PT} for simplicity. 
\begin{thm}\label{thm:3b}
{Consider the Bob-with-input scenario.} Let $\{\sigma_{a|xy}\}$ be a quantum realisable assemblage. Then, there exist CPTP maps $\{F^y\}_{y>1}$ and a dilation  $\widetilde{\sigma}_{a|x1}$ of $\sigma_{a|x1}$ for each $(a,x)$ through an auxiliary system $E$, such that:
\begin{align*}
\begin{cases}
\sigma_{a|xy} = F^y[\widetilde{\sigma}_{a|x1}] \quad \forall \, y>1 \,,\\
\sigma_{a|x1} = \mathrm{tr}_{E'} \left\{ \widetilde{\sigma}_{a|x1} \right\}\,.
\end{cases}
\end{align*}
\end{thm}
\begin{proof}
Let $\rho$, $\{M_{a|x}\}$ and $\{\mathcal{E}_y\}$ be the state, measurements and CPTP maps that provide a quantum realisation of $\{\sigma_{a|xy}\}$. In diagrammatic notation:
\begin{align*}
\sigma_{a|xy} \quad = \quad %
\begin{tikzpicture}
	\begin{pgfonlayer}{nodelayer}
		\node [style=none] (0) at (0, -0.5) {$\rho$};
		\node [style=none] (1) at (-1.5, 0) {};
		\node [style=none] (2) at (0, -1.25) {};
		\node [style=none] (3) at (1.5, 0) {};
		\node [style=none] (4) at (-1, 1) {\tiny{$M_{a|x}$}};
		\node [style=none] (5) at (-1, 0) {};
		\node [style=small box] (6) at (1, 1.25) {$\mathcal{E}_{y}$};
		\node [style=none] (7) at (1, 0) {};
		\node [style=none] (8) at (1, 2.25) {};
		\node [style=right label] (9) at (-1, 0.25) {$A$};
		\node [style=right label] (10) at (1, 0.25) {$B_{i}$};
		\node [style=right label] (11) at (1, 2) {$B$};
		\node [style=none] (12) at (-2, 0.75) {};
		\node [style=none] (13) at (0, 0.75) {};
		\node [style=none] (14) at (-1, 1.75) {};
		\node [style=none] (15) at (-1, 0.75) {};
	\end{pgfonlayer}
	\begin{pgfonlayer}{edgelayer}
		\draw (1.center) to (2.center);
		\draw (2.center) to (3.center);
		\draw (3.center) to (1.center);
		\draw [qWire] (6) to (7.center);
		\draw [qWire] (8.center) to (6);
		\draw (14.center) to (12.center);
		\draw (12.center) to (13.center);
		\draw (13.center) to (14.center);
		\draw [qWire] (15.center) to (5.center);
	\end{pgfonlayer}
\end{tikzpicture}}.
\end{align*}
Let the auxiliary system $E$ on state $\ket{\chi}$ and the unitary operators $\{U_y\}$ provide a unitary dilation for the CPTP maps $\{\mathcal{E}_y\}$, namely:
\begin{align*}
\sigma_{a|xy} \quad = \quad %
\begin{tikzpicture}
	\begin{pgfonlayer}{nodelayer}
		\node [style=none] (0) at (0, -0.5) {$\rho$};
		\node [style=none] (1) at (-1.5, 0) {};
		\node [style=none] (2) at (0, -1.25) {};
		\node [style=none] (3) at (1.5, 0) {};
		\node [style=none] (4) at (-1, 1) {\tiny{$M_{a|x}$}};
		\node [style=none] (5) at (-1, 0) {};
		\node [style=none] (7) at (1, 0) {};
		\node [style=none] (8) at (1, 2.25) {};
		\node [style=right label] (9) at (-1, 0.25) {$A$};
		\node [style=right label] (10) at (1, 0.25) {$B_{i}$};
		\node [style=right label] (11) at (1, 2) {$B$};
		\node [style=none] (12) at (-2, 0.75) {};
		\node [style=none] (13) at (0, 0.75) {};
		\node [style=none] (14) at (-1, 1.75) {};
		\node [style=none] (15) at (-1, 0.75) {};
		\node [style=none] (16) at (2.75, -0.5) {$\chi$};
		\node [style=none] (17) at (2, 0) {};
		\node [style=none] (18) at (3.5, 0) {};
		\node [style=none] (19) at (2.75, -1.25) {};
		\node [style=none] (20) at (0.5, 0.75) {};
		\node [style=none] (21) at (3.5, 0.75) {};
		\node [style=none] (22) at (3.5, 1.5) {};
		\node [style=none] (23) at (0.5, 1.5) {};
		\node [style=none] (24) at (1, 1.5) {};
		\node [style=none] (25) at (1, 0.75) {};
		\node [style=none] (26) at (2.75, 0) {};
		\node [style=none] (27) at (2.75, 0.75) {};
		\node [style=none] (28) at (2.75, 1.5) {};
		\node [style=upground] (29) at (2.75, 2.5) {};
		\node [style=right label] (30) at (2.75, 0.25) {$E$};
		\node [style=none] (31) at (2, 1) {\footnotesize{$U_{y}$}};
		\node [style=right label] (32) at (2.75, 1.75) {$E'$};
	\end{pgfonlayer}
	\begin{pgfonlayer}{edgelayer}
		\draw (1.center) to (2.center);
		\draw (2.center) to (3.center);
		\draw (3.center) to (1.center);
		\draw (14.center) to (12.center);
		\draw (12.center) to (13.center);
		\draw (13.center) to (14.center);
		\draw [qWire] (15.center) to (5.center);
		\draw (17.center) to (19.center);
		\draw (19.center) to (18.center);
		\draw (17.center) to (18.center);
		\draw (23.center) to (20.center);
		\draw (20.center) to (21.center);
		\draw (21.center) to (22.center);
		\draw (22.center) to (23.center);
		\draw [qWire] (8.center) to (24.center);
		\draw [qWire] (25.center) to (7.center);
		\draw [qWire] (29) to (28.center);
		\draw [qWire] (27.center) to (26.center);
	\end{pgfonlayer}
\end{tikzpicture}}.
\end{align*}
Let us now define the maps $\{F_y\}$ as follows:
\begin{align*}
\begin{tikzpicture}
	\begin{pgfonlayer}{nodelayer}
		\node [style=none] (29) at (1, 0.5) {};
		\node [style=none] (30) at (1, -0.5) {};
		\node [style=none] (31) at (-1, -0.5) {};
		\node [style=none] (32) at (-1, 0.5) {};
		\node [style=none] (35) at (0, 0.5) {};
		\node [style=none] (36) at (0, 1.5) {};
		\node [style=none] (42) at (0.5, -1.5) {};
		\node [style=none] (43) at (-0.5, -0.5) {};
		\node [style=none] (44) at (-0.5, -1.5) {};
		\node [style=none] (45) at (0.5, -0.5) {};
		\node [style=none] (48) at (0, 0) {$F_y$};
		\node [style=none] (56) at (0, 0.75) {};
		\node [style=right label] (57) at (0, 0.75) {$B$};
		\node [style=right label] (58) at (-0.5, -1.5) {$B$};
		\node [style=right label] (59) at (0.5, -1.5) {$E'$};
	\end{pgfonlayer}
	\begin{pgfonlayer}{edgelayer}
		\draw (32.center) to (29.center);
		\draw (29.center) to (30.center);
		\draw (30.center) to (31.center);
		\draw (31.center) to (32.center);
		\draw [qWire] (36.center) to (35.center);
		\draw [qWire] (45.center) to (42.center);
		\draw [qWire] (43.center) to (44.center);
	\end{pgfonlayer}
\end{tikzpicture}} \quad := \quad %
\begin{tikzpicture}
	\begin{pgfonlayer}{nodelayer}
		\node [style=none] (25) at (1.5, -0.25) {};
		\node [style=none] (26) at (1.5, -1.25) {};
		\node [style=none] (27) at (-1.5, -1.25) {};
		\node [style=none] (28) at (-1.5, -0.25) {};
		\node [style=none] (29) at (1.5, 1.25) {};
		\node [style=none] (30) at (1.5, 0.25) {};
		\node [style=none] (31) at (-1.5, 0.25) {};
		\node [style=none] (32) at (-1.5, 1.25) {};
		\node [style=none] (33) at (1, 1.25) {};
		\node [style=none] (34) at (1, 2) {};
		\node [style=none] (35) at (-1, 1.25) {};
		\node [style=none] (36) at (-1, 2.5) {};
		\node [style=upground] (37) at (1, 2.25) {};
		\node [style=none] (42) at (1, -0.25) {};
		\node [style=none] (43) at (-1, 0.25) {};
		\node [style=none] (44) at (-1, -0.25) {};
		\node [style=none] (45) at (1, 0.25) {};
		\node [style=none] (47) at (0, -0.75) {$U_1^{-1}$};
		\node [style=none] (48) at (0, 0.75) {$U_y$};
		\node [style=none] (49) at (1, -1.75) {};
		\node [style=none] (50) at (-1, -1.25) {};
		\node [style=none] (51) at (-1, -1.75) {};
		\node [style=none] (52) at (1, -1.25) {};
		\node [style=right label] (53) at (-1, -1.75) {$B$};
		\node [style=right label] (54) at (1, -1.75) {$E'$};
		\node [style=none] (56) at (-1, 1.5) {};
		\node [style=right label] (57) at (-1, 1.5) {$B$};
		\node [style=right label] (58) at (1, 1.5) {$E'$};
	\end{pgfonlayer}
	\begin{pgfonlayer}{edgelayer}
		\draw (28.center) to (25.center);
		\draw (25.center) to (26.center);
		\draw (26.center) to (27.center);
		\draw (27.center) to (28.center);
		\draw (32.center) to (29.center);
		\draw (29.center) to (30.center);
		\draw (30.center) to (31.center);
		\draw (31.center) to (32.center);
		\draw [qWire] (33.center) to (34.center);
		\draw [qWire] (36.center) to (35.center);
		\draw [qWire] (45.center) to (42.center);
		\draw [qWire] (43.center) to (44.center);
		\draw [qWire] (52.center) to (49.center);
		\draw [qWire] (50.center) to (51.center);
	\end{pgfonlayer}
\end{tikzpicture}}\,.
\end{align*}
Notice that each of these $F_y$ is a CPTP map acting on the joint system $BE'$. 

Let us now show the first part of the claim, i.e.~$\sigma_{a|xy} = F^y[\widetilde{\sigma}_{a|x1}] \quad \forall \, y>1$:
\begin{align*}
F_y[\tilde{\sigma}_{a|x1}]\quad &= \quad %
\begin{tikzpicture}
	\begin{pgfonlayer}{nodelayer}
		\node [style=none] (0) at (0, -0.5) {$\rho$};
		\node [style=none] (1) at (-1.5, 0) {};
		\node [style=none] (2) at (0, -1.25) {};
		\node [style=none] (3) at (1.5, 0) {};
		\node [style=none] (4) at (-1, 1) {\tiny{$M_{a|x}$}};
		\node [style=none] (5) at (-1, 0) {};
		\node [style=none] (6) at (1, 0.75) {};
		\node [style=none] (7) at (1, 0) {};
		\node [style=right label] (8) at (-1, 0.25) {$A$};
		\node [style=right label] (9) at (1, 0.25) {$B_{i}$};
		\node [style=none] (10) at (-2, 0.75) {};
		\node [style=none] (11) at (0, 0.75) {};
		\node [style=none] (12) at (-1, 1.75) {};
		\node [style=none] (13) at (-1, 0.75) {};
		\node [style=none] (14) at (3, 0.75) {};
		\node [style=none] (15) at (2, 0) {};
		\node [style=right label] (16) at (3, 0.25) {$E$};
		\node [style=none] (17) at (3, -1) {};
		\node [style=none] (18) at (3, -0.5) {$\chi$};
		\node [style=none] (19) at (3, 0) {};
		\node [style=none] (20) at (4, 0) {};
		\node [style=none] (21) at (0.5, 1.75) {};
		\node [style=none] (22) at (3.5, 1.75) {};
		\node [style=none] (23) at (3.5, 0.75) {};
		\node [style=none] (24) at (0.5, 0.75) {};
		\node [style=none] (25) at (3.5, 4) {};
		\node [style=none] (26) at (3.5, 3) {};
		\node [style=none] (27) at (0.5, 3) {};
		\node [style=none] (28) at (0.5, 4) {};
		\node [style=none] (29) at (3.5, 5.5) {};
		\node [style=none] (30) at (3.5, 4.5) {};
		\node [style=none] (31) at (0.5, 4.5) {};
		\node [style=none] (32) at (0.5, 5.5) {};
		\node [style=none] (33) at (3, 5.75) {};
		\node [style=none] (34) at (3, 5.5) {};
		\node [style=none] (35) at (1, 5.5) {};
		\node [style=none] (36) at (1, 6.5) {};
		\node [style=upground] (37) at (3, 6) {};
		\node [style=none] (38) at (3, 3) {};
		\node [style=none] (39) at (3, 1.75) {};
		\node [style=none] (40) at (1, 3) {};
		\node [style=none] (41) at (1, 1.75) {};
		\node [style=none] (42) at (3, 4) {};
		\node [style=none] (43) at (1, 4.5) {};
		\node [style=none] (44) at (1, 4) {};
		\node [style=none] (45) at (3, 4.5) {};
		\node [style=none] (46) at (2, 1.25) {$U_1$};
		\node [style=none] (47) at (2, 3.5) {$U_1^{-1}$};
		\node [style=none] (48) at (2, 5) {$U_y$};
		\node [style=none] (49) at (-2.5, 2.5) {};
		\node [style=none] (50) at (4.5, 2.5) {};
		\node [style=none] (51) at (4.5, -1.5) {};
		\node [style=none] (52) at (-1, -1.5) {};
		\node [style=none] (53) at (-2.5, -0.5) {};
	\end{pgfonlayer}
	\begin{pgfonlayer}{edgelayer}
		\draw (1.center) to (2.center);
		\draw (2.center) to (3.center);
		\draw (3.center) to (1.center);
		\draw [qWire] (6.center) to (7.center);
		\draw (12.center) to (10.center);
		\draw (10.center) to (11.center);
		\draw (11.center) to (12.center);
		\draw [qWire] (13.center) to (5.center);
		\draw (17.center) to (15.center);
		\draw (15.center) to (20.center);
		\draw (20.center) to (17.center);
		\draw [qWire] (19.center) to (14.center);
		\draw (21.center) to (22.center);
		\draw (22.center) to (23.center);
		\draw (23.center) to (24.center);
		\draw (24.center) to (21.center);
		\draw (28.center) to (25.center);
		\draw (25.center) to (26.center);
		\draw (26.center) to (27.center);
		\draw (27.center) to (28.center);
		\draw (32.center) to (29.center);
		\draw (29.center) to (30.center);
		\draw (30.center) to (31.center);
		\draw (31.center) to (32.center);
		\draw [qWire] (33.center) to (34.center);
		\draw [qWire] (36.center) to (35.center);
		\draw [qWire] (38.center) to (39.center);
		\draw [qWire] (40.center) to (41.center);
		\draw [qWire] (45.center) to (42.center);
		\draw [qWire] (43.center) to (44.center);
		\draw [thick gray dashed edge] (49.center) to (50.center);
		\draw [thick gray dashed edge] (50.center) to (51.center);
		\draw [thick gray dashed edge] (51.center) to (52.center);
		\draw [thick gray dashed edge] (52.center) to (53.center);
		\draw [thick gray dashed edge] (53.center) to (49.center);
	\end{pgfonlayer}
\end{tikzpicture}}
\quad =\quad %
\begin{tikzpicture}
	\begin{pgfonlayer}{nodelayer}
		\node [style=none] (0) at (0, -0.5) {$\rho$};
		\node [style=none] (1) at (-1.5, 0) {};
		\node [style=none] (2) at (0, -1.25) {};
		\node [style=none] (3) at (1.5, 0) {};
		\node [style=none] (4) at (-1, 1) {\tiny{$M_{a|x}$}};
		\node [style=none] (5) at (-1, 0) {};
		\node [style=none] (6) at (1, 0.75) {};
		\node [style=none] (7) at (1, 0) {};
		\node [style=right label] (8) at (-1, 0.25) {$A$};
		\node [style=right label] (9) at (1, 0.25) {$B_{i}$};
		\node [style=none] (10) at (-2, 0.75) {};
		\node [style=none] (11) at (0, 0.75) {};
		\node [style=none] (12) at (-1, 1.75) {};
		\node [style=none] (13) at (-1, 0.75) {};
		\node [style=none] (14) at (3, 0.75) {};
		\node [style=none] (15) at (2, 0) {};
		\node [style=right label] (16) at (3, 0.25) {$E$};
		\node [style=none] (17) at (3, -1) {};
		\node [style=none] (18) at (3, -0.5) {$\chi$};
		\node [style=none] (19) at (3, 0) {};
		\node [style=none] (20) at (4, 0) {};
		\node [style=none] (21) at (3.5, 1.75) {};
		\node [style=none] (22) at (3.5, 0.75) {};
		\node [style=none] (23) at (0.5, 0.75) {};
		\node [style=none] (24) at (0.5, 1.75) {};
		\node [style=none] (25) at (3, 2) {};
		\node [style=none] (26) at (3, 1.75) {};
		\node [style=none] (27) at (1, 1.75) {};
		\node [style=none] (28) at (1, 2.75) {};
		\node [style=upground] (29) at (3, 2.25) {};
		\node [style=none] (30) at (1, 0.75) {};
		\node [style=none] (31) at (3, 0.75) {};
		\node [style=none] (32) at (2, 1.25) {$U_y$};
	\end{pgfonlayer}
	\begin{pgfonlayer}{edgelayer}
		\draw (1.center) to (2.center);
		\draw (2.center) to (3.center);
		\draw (3.center) to (1.center);
		\draw [qWire] (6.center) to (7.center);
		\draw (12.center) to (10.center);
		\draw (10.center) to (11.center);
		\draw (11.center) to (12.center);
		\draw [qWire] (13.center) to (5.center);
		\draw (17.center) to (15.center);
		\draw (15.center) to (20.center);
		\draw (20.center) to (17.center);
		\draw [qWire] (19.center) to (14.center);
		\draw (24.center) to (21.center);
		\draw (21.center) to (22.center);
		\draw (22.center) to (23.center);
		\draw (23.center) to (24.center);
		\draw [qWire] (25.center) to (26.center);
		\draw [qWire] (28.center) to (27.center);
	\end{pgfonlayer}
\end{tikzpicture}} \\
\quad &=\quad %
\begin{tikzpicture}
	\begin{pgfonlayer}{nodelayer}
		\node [style=none] (0) at (0, -0.5) {$\rho$};
		\node [style=none] (1) at (-1.5, 0) {};
		\node [style=none] (2) at (0, -1.25) {};
		\node [style=none] (3) at (1.5, 0) {};
		\node [style=none] (4) at (-1, 1) {\tiny{$M_{a|x}$}};
		\node [style=none] (5) at (-1, 0) {};
		\node [style=small box] (6) at (1, 1.25) {$\mathcal{E}_{y}$};
		\node [style=none] (7) at (1, 0) {};
		\node [style=none] (8) at (1, 2.25) {};
		\node [style=right label] (9) at (-1, 0.25) {$A$};
		\node [style=right label] (10) at (1, 0.25) {$B_{i}$};
		\node [style=right label] (11) at (1, 2) {$B$};
		\node [style=none] (12) at (-2, 0.75) {};
		\node [style=none] (13) at (0, 0.75) {};
		\node [style=none] (14) at (-1, 1.75) {};
		\node [style=none] (15) at (-1, 0.75) {};
	\end{pgfonlayer}
	\begin{pgfonlayer}{edgelayer}
		\draw (1.center) to (2.center);
		\draw (2.center) to (3.center);
		\draw (3.center) to (1.center);
		\draw [qWire] (6) to (7.center);
		\draw [qWire] (8.center) to (6);
		\draw (14.center) to (12.center);
		\draw (12.center) to (13.center);
		\draw (13.center) to (14.center);
		\draw [qWire] (15.center) to (5.center);
	\end{pgfonlayer}
\end{tikzpicture}} \quad = \quad \sigma_{a|xy}\,.
\end{align*}
For the second part of the claim, notice that
\begin{align*}
\mathrm{tr}_{E'} \left\{ \widetilde{\sigma}_{a|x1} \right\} \quad = \quad %
\begin{tikzpicture}
	\begin{pgfonlayer}{nodelayer}
		\node [style=none] (0) at (0, -0.5) {$\rho$};
		\node [style=none] (1) at (-1.5, 0) {};
		\node [style=none] (2) at (0, -1.25) {};
		\node [style=none] (3) at (1.5, 0) {};
		\node [style=none] (4) at (-1, 1) {\tiny{$M_{a|x}$}};
		\node [style=none] (5) at (-1, 0) {};
		\node [style=none] (6) at (1, 0.75) {};
		\node [style=none] (7) at (1, 0) {};
		\node [style=right label] (8) at (-1, 0.25) {$A$};
		\node [style=right label] (9) at (1, 0.25) {$B_{i}$};
		\node [style=none] (10) at (-2, 0.75) {};
		\node [style=none] (11) at (0, 0.75) {};
		\node [style=none] (12) at (-1, 1.75) {};
		\node [style=none] (13) at (-1, 0.75) {};
		\node [style=none] (14) at (3, 0.75) {};
		\node [style=none] (15) at (2, 0) {};
		\node [style=right label] (16) at (3, 0.25) {$E$};
		\node [style=none] (17) at (3, -1) {};
		\node [style=none] (18) at (3, -0.25) {$\chi$};
		\node [style=none] (19) at (3, 0) {};
		\node [style=none] (20) at (4, 0) {};
		\node [style=none] (21) at (3.5, 1.75) {};
		\node [style=none] (22) at (3.5, 0.75) {};
		\node [style=none] (23) at (0.5, 0.75) {};
		\node [style=none] (24) at (0.5, 1.75) {};
		\node [style=none] (25) at (3, 2) {};
		\node [style=none] (26) at (3, 1.75) {};
		\node [style=none] (27) at (1, 1.75) {};
		\node [style=none] (28) at (1, 2.75) {};
		\node [style=upground] (29) at (3, 2.25) {};
		\node [style=none] (30) at (1, 0.75) {};
		\node [style=none] (31) at (3, 0.75) {};
		\node [style=none] (32) at (2, 1.25) {$U_1$};
		\node [style=right label] (33) at (1, 2.25) {$B$};
	\end{pgfonlayer}
	\begin{pgfonlayer}{edgelayer}
		\draw (1.center) to (2.center);
		\draw (2.center) to (3.center);
		\draw (3.center) to (1.center);
		\draw [qWire] (6.center) to (7.center);
		\draw (12.center) to (10.center);
		\draw (10.center) to (11.center);
		\draw (11.center) to (12.center);
		\draw [qWire] (13.center) to (5.center);
		\draw (17.center) to (15.center);
		\draw (15.center) to (20.center);
		\draw (20.center) to (17.center);
		\draw [qWire] (19.center) to (14.center);
		\draw (24.center) to (21.center);
		\draw (21.center) to (22.center);
		\draw (22.center) to (23.center);
		\draw (23.center) to (24.center);
		\draw [qWire] (25.center) to (26.center);
		\draw [qWire] (28.center) to (27.center);
	\end{pgfonlayer}
\end{tikzpicture}}
\quad =\quad %
\begin{tikzpicture}
	\begin{pgfonlayer}{nodelayer}
		\node [style=none] (0) at (0, -0.5) {$\rho$};
		\node [style=none] (1) at (-1.5, 0) {};
		\node [style=none] (2) at (0, -1.25) {};
		\node [style=none] (3) at (1.5, 0) {};
		\node [style=none] (4) at (-1, 1) {\tiny{$M_{a|x}$}};
		\node [style=none] (5) at (-1, 0) {};
		\node [style=small box] (6) at (1, 1.25) {$\mathcal{E}_{y}$};
		\node [style=none] (7) at (1, 0) {};
		\node [style=none] (8) at (1, 2.25) {};
		\node [style=right label] (9) at (-1, 0.25) {$A$};
		\node [style=right label] (10) at (1, 0.25) {$B_{i}$};
		\node [style=right label] (11) at (1, 2) {$B$};
		\node [style=none] (12) at (-2, 0.75) {};
		\node [style=none] (13) at (0, 0.75) {};
		\node [style=none] (14) at (-1, 1.75) {};
		\node [style=none] (15) at (-1, 0.75) {};
	\end{pgfonlayer}
	\begin{pgfonlayer}{edgelayer}
		\draw (1.center) to (2.center);
		\draw (2.center) to (3.center);
		\draw (3.center) to (1.center);
		\draw [qWire] (6) to (7.center);
		\draw [qWire] (8.center) to (6);
		\draw (14.center) to (12.center);
		\draw (12.center) to (13.center);
		\draw (13.center) to (14.center);
		\draw [qWire] (15.center) to (5.center);
	\end{pgfonlayer}
\end{tikzpicture}} \quad = \quad \sigma_{a|x1}\,,
\end{align*}
which concludes the proof of the theorem. 
\end{proof}
We can now prove Lemma \ref{theIDthing} as a corollary of Theorem \ref{thm:3b}:
\begin{proof}[Proof of Lemma \ref{theIDthing}]
From Theorem \ref{thm:3b} we know there exists a purification $\widetilde{\sigma}_{a|x1} $ of ${\sigma}_{a|x1} $ for each $(a,x)$, and CPTP maps $F_y$ for each $y>1$, such that $\sigma_{a|xy} = F^y[\widetilde{\sigma}_{a|x1}] \quad \forall \, y>1$. 

Since ${\sigma}_{a|x1} $ is pure for each $(a,x)$, then $\widetilde{\sigma}_{a|x1} = {\sigma}_{a|x1} \otimes \omega $, where $\omega$ is a fixed state of the auxiliary system $E$. 

Hence, $\sigma_{a|xy} = F^y[{\sigma}_{a|x1} \otimes \omega] = \mathcal{F}^y[{\sigma}_{a|x1}]$, where the operators $\mathcal{F}_y$ are CPTP since the $F_y$'s are. This concludes the proof of the Lemma.
\end{proof}

The alternative proof that the assemblage of Eq.~\eqref{e:example} is post-quantum is then the following: 
\begin{proof}
Now we will show that even though the assemblage $\{\sigma^*_{a|xy}\}$ is such that $\{\sigma^*_{a|x1}\}$ is a collection of pure quantum states, it does not comply with Lemma \ref{theIDthing}, and hence has no quantum realisation. For this, let $X,Y$ and $Z$ be the Pauli matrices, and notice that: {
\begin{align*}
\sigma^*_{0|10} - \sigma^*_{1|10} &= \tfrac{1}{2} X \,, \quad \sigma^*_{0|11} - \sigma^*_{1|11 }= \tfrac{1}{2} X \,, \\ 
\sigma^*_{0|20} - \sigma^*_{1|20} &= \tfrac{1}{2} Y \,, \quad \sigma^*_{0|21} - \sigma^*_{1|21} = - \tfrac{1}{2} Y \,, \\
\sigma^*_{0|30} - \sigma^*_{1|30} &= \tfrac{1}{2} Z \,,  \quad \sigma^*_{0|31} - \sigma^*_{1|31} = \tfrac{1}{2} Z \,. 
\end{align*}}
Should $\{\sigma^*_{a|xy}\}$ have a quantum realisation, then by Lemma \ref{theIDthing} there exists a CPTP map $\Lambda$ such that
$\Lambda[\id] = \id\,,\, \Lambda[X]=X\,,\, \Lambda[Z]=Z\,,\, \Lambda[Y]=-Y\,$. However, if the map $\Lambda$ is applied to one half of the maximally-entangled state $\rho = \tfrac{1}{4} \left( \id \otimes \id - X \otimes X + Y \otimes Y - Z \otimes Z \right)$ one gets a non-positive matrix, which shows that $\Lambda$ is actually not CPTP. This contradiction shows that $\{\sigma^*_{a|xy}\}$ is indeed a post-quantum assemblage. 
\end{proof}

\section{Steering inequalities and quantum bounds for the Instrumental steering scenario}

In this section we discuss how to use steering inequalities to certify the post-quantumness of assemblages in the Instrumental steering scenario. 

The method we construct relies fundamentally on the connection between quantum correlations in the Instrumental scenario, and the NPA-hierarchy of correlations in Bell scenarios \cite{pironio}. In a nutshell, a result in Ref.~\cite{pironio} says that ``a correlation $p(a|xy)$ in the Instrumental scenario admits a quantum realisation if it may arise as the post-selection $p(ab|xa)$ of a correlation $p(ab|xy)$ in a Bell scenario, which satisfies all the levels of the NPA hierarchy''. This, together with Def.~\ref{def:steeint}, make the following observation: 
\begin{remk}\label{rem:aqi}
For every quantum assemblage $\{\sigma_{a|x}\}$ in the Instrumental scenario, there exists a $\aQ$ assemblage $\{\sigma_{a|xy}\}$ in the Bob-with-Input scenario, such that $\sigma_{a|x} = \sigma_{a|xa}$. 
\end{remk}

A linear steering inequality in the Instrumental steering scenario is defined by the functional:
\begin{align}
\beta\left(\{\sigma_{a|x}\} \right) = \Tr \left\{ \sum_{a,x} F_{ax} \, \sigma_{a|x} \right\}\,,
\end{align}
where $\{F_{ax}\}$ is a set of Hermitian operators. 

Defining the quantum bound as: 
\begin{align}
\beta^{\mathrm{Q}} := \min_{\{\sigma_{a|x}\} \, \text{is quantum}} \Tr\left\{ \sum_{a,x} F_{ax} \, \sigma_{a|x}   \right\}\,,
\end{align}
a quantum-steering inequality to certify post-quantumness of assemblages in the Instrumental steering scenario reads: 
\begin{align}
\beta^{\mathrm{Q}} \leq \Tr \left\{ \sum_{a,x} F_{ax} \, \sigma_{a|x} \right\}\,.
\end{align}
Given that the set of quantum assemblages in the Instrumental steering scenario is not efficiently characterised, one may not always be able to compute the value of $\beta^{\mathrm{Q}}$. Hence, we present a method to lower-bound  its value, based on remark \ref{rem:aqi}. First, define the following quantity:
\begin{align}
\beta^{\aQ_I} := \min_{\{\sigma_{a|xy}\} \in \aQ} \Tr\left\{ \sum_{a,x} F_{ax} \, \sigma_{a|xa}   \right\}\,.
\end{align}
It follows that $\beta^{\aQ_I} \leq \beta^{\mathrm{Q}}$. Hence, if an assemblage gives a value for the steering functional which is lower than  $\beta^{\aQ_I}$, the assemblage is not quantum. 

\section{Proof of Theorems {\ref{theo:inst1} and \ref{thm:inst2}}: {post-quantum steering exists in the Instrumental steering scenario, and does not imply post-quantum instrumental correlations}}

Here we present the proof of {Theorem \ref{thm:inst2}}, since {Theorem \ref{theo:inst1}} is implied by it. Here we {recall} the theorem below with its proof. 
\\
\\
\textbf{{Theorem \ref{thm:inst2}}.} \textit{The following assemblage in the instrumental steering scenario has no quantum realisation, and may only yield quantum correlations in the traditional instrumental setup.}{
\begin{align*}
&\sigma^*_{0|1} = \tfrac{1}{2} (\ket{0}+\ket{1})(\bra{0}+\bra{1})\,, \\
&\sigma^*_{0|2} = \tfrac{1}{2} (\ket{0}- i \,  \ket{1})(\bra{0}+i \,  \bra{1})\,, \\
&\sigma^*_{0|3} = \tfrac{1}{2} \ket{0}\bra{0}\,, \quad  \sigma^*_{1|3} = \tfrac{1}{2} \left(\ket{1}\bra{1}\right)^{\mathrm{T}}\,,\\
&\sigma^*_{1|1} = \tfrac{1}{2} \left((\ket{0}-\ket{1})(\bra{0}-\bra{1})\right)^{\mathrm{T}}\,, \\
&\sigma^*_{1|2} = \tfrac{1}{2} \left((\ket{0}+i \,\ket{1})(\bra{0}-i \, \bra{1})\right)^{\mathrm{T}}\,.
\end{align*}}
\begin{proof}
First notice that the assemblage $\{\sigma^*_{a|x}\}$, where $x \in \{0,1,2\}$ and $a \in \{0,1\}$, may {be mathematically expressed as one where} Alice and Bob share a maximally entangled state of two qubits, Alice chooses between the three Pauli measurements to perform on her qubit, and Bob applies on his qubit the identity channel when $a=0$ and the transpose when $a=1$. Hence, this assemblage may be mathematically obtained from that in {Theorem \ref{thm:ex}} by setting $\sigma^*_{a|x} = \sigma^*_{a|xa}$. This shows that $\{\sigma^*_{a|x}\}$ is indeed a valid assemblage in the scenario, by {Def.~\ref{def:steeint}}. 

Let us now prove that the assemblage is post-quantum. Let us assume, for contradiction, that $\{\sigma^*_{a|x}\}$ has a quantum realisation. That is, assume there is a state $\rho$, POVMs $\{M_{a|x}\}$ for Alice, and CPTP maps $\{\mathcal{E}^a\}$ for Bob, such that $\sigma^*_{a|x} = \mathcal{E}^a \left[ \mathrm{tr}_{A}\left\{ M_{a|x} \otimes \id  \, \rho   \right\} \right]$. 
A self-testing argument, imposes that $\{M_{a|x}\}$ be Pauli measurements and $\rho$ a maximally entangled state, up to a local isometry. It hence follows that the assemblage $\{\sigma^*_{a|xy}\}$ of {Theorem \ref{thm:ex}} may be expressed as $\sigma^*_{a|xy} \equiv  \mathcal{E}^y \left[ \mathrm{tr}_{A}\left\{ M_{a|x} \otimes \id  \, \rho   \right\} \right]$. This is, however, impossible, since $\{\sigma^*_{a|xy}\}$ has no quantum realisation. This proves that a quantum model for $\{\sigma^*_{a|x}\}$ cannot exist. 

Finally, let us show that the correlations $p(ab|x) =  \Tr\left\{ N_b \sigma_{a|x} \right\}$ have a quantum realisation in the instrumental scenario for all POVM $\{N_b\}$. For a given choice of POVM $\{N_b\}$, consider the correlations $p(ab|xy)$ in a Bell scenario given by $p(ab|xy) = \Tr\left\{ N_b \, \sigma^*_{a|xy} \right\}$, where  $\{\sigma^*_{a|xy}\}$ is the assemblage of {Theorem \ref{thm:ex}}. On the one hand, notice that $p(ab|x) \equiv p(ab|xa)$. On the other hand, $\{\sigma^*_{a|xy}\}$ is a PTP assemblage, and by Theorem \ref{PTPquantum} the correlations $p(ab|xy)$ have a quantum realisation. This quantum model for $p(ab|xy)$ then gives a quantum model for $p(ab|x)$, and the claim follows. 
\end{proof}


\begin{thebibliography}{99}
\bibitem{footnote1} Note that the fact that the assemblage can be represented as arising through the application of a postitive but not completely positive map is a useful mathematical fact, which simultaneously provides a route to show that it is post-quantum (since it involves a map which is not allowed in quantum theory), {and} that it cannot lead to post-quantum nonlocality, due to Ref.~\cite{oform}. This representation does not dictate that the assemblage has to arise in a such a way; how it {would be physically prepared} will depend upon details of the post-quantum theory. {Our particular way to express the assemblage here} is merely one {mathematical} representation that is useful for our purposes.
\bibitem{FTN}
{Note that the second measurement $\{M_{a_2|x_2}\}_{a_2}$ will in general implicitly depend upon both $a_1$ and $x_1$ of the prior measurement -- we can assume without loss of generality that the post-measurement state after the first measurement contains within it all the classical information about the choice of measurement and the outcome, such that the second measurement implicitly conditions upon this classical information. This means we do not need to explicitly write the dependence upon $a_1$ and $x_1$ in $M_{a_2|x_2}$. }
\end{thebibliography}
\end{document}